\documentclass{article}
\usepackage[utf8]{inputenc}
\usepackage{graphicx}
\usepackage{xcolor}
\usepackage{float}
\usepackage{amsthm}
\usepackage{hyperref}
\usepackage{amsmath}
\usepackage{cleveref}
\usepackage{verbatim}
\usepackage{amssymb}
\usepackage{cleveref}

\usepackage[boxed,noend,ruled,linesnumbered]{algorithm2e}

\SetCommentSty{mycommfont}
\newtheorem{theorem}{Theorem}[]
\newtheorem{corollary}{Corollary}[]
\newtheorem{lemma}{Lemma}
\crefname{lemma}{Lemma}{Lemmas}
\crefname{theorem}{Theorem}{Theorems}

\theoremstyle{definition}
\newtheorem{defn}[]{Definition} 

\SetKwInOut{Input}{Input}
\SetKwInOut{DataStructure}{Data Structure}
\SetKwInOut{Output}{Output\,}
\SetKwInOut{Data}{Data}
\SetKwProg{Tree}{Tree}{}{EndTree}

\title{The Algorithmic Complexity of Tree-Clique Width}
\author{Chris Aronis\footnote{aronis.ch@gmail.com}}
\date{}

\begin{document}

\maketitle

\begin{abstract}

Tree-width has been proven to be a useful parameter to design fast and efficient algorithms for intractable problems. However, while tree-width is low on relatively sparse graphs can be arbitrary high on dense graphs. Therefore, we introduce tree-clique width, denoted by $tcl(G)$ for a graph $G$,  a new width measure for tree decompositions. The main aim of such a parameter is to extend the algorithmic gains of tree-width on more structured and dense graphs. In this paper, we show that tree-clique width is NP-complete and that there is no constant factor approximation algorithm for any constant value $c$. We also provide algorithms to compute tree-clique width for general graphs and for special graphs such as cographs and permutation graphs. We seek to understand further tree-clique width and its properties and to research whether it can be used as an alternative where tree-width fails.

\end{abstract}

\section{Introduction}

Being able to efficiently solve intractable problems, namely NP-hard problems, has been one of the major goals for researchers. Unfortunately, they soon realized it is quite unlikely that these problems may admit a polynomial time algorithm, unless P=NP. Therefore, attention has been drawn to special graph cases where all these hard problems might be easier to tackle. It has become apparent that one of those cases are the trees. Numerous NP-complete problems become tractable on trees and admit fast polynomial time algorithms. One of the main reasons for that has to do with their structure. Namely, trees have two well known separation properties. The first one results when we delete a vertex $v$ of a tree along with all its incident edges. Then we know that the tree falls apart into disconnected parts equal to the degree of $v$. The second property comes when we delete an edge of a tree which creates exactly two disconnected components. These two properties are very convenient because we can create subproblems which have very little interaction between each other. We can then exploit the existence of these subproblems by solving them using dynamic programming or any other similar algorithmic technique.

Having said that, it is obvious that we could benefit by having treelike input graphs. Unfortunately, most of the time this is not the case as in very many real world graphs \cite{maniu2019experimental}. So it would be a good idea to associate tree structures to graphs. Otherwise, we would like to know how hard it is to do so, in other words, how many steps away from, a graph $G$ is in becoming a tree. All these ideas can be captured by a graph parameter called tree-width which measures how well a graph resembles a tree.

Researchers have proposed different notions to generalize this idea of tree-width by introducing other width measures as well as other graph decomposition methods. Some of them are the rank-width, tree-length, clique-width, tree-depth and many others. The benefit they offer is that they can be used as bounds for different kinds of problems and handle different kinds of graphs, such as dense graphs. For example, if we were to use the maximum over all bags of a tree decomposition of the number of cliques needed to cover all vertices in the bag as measure of width (of a tree decomposition), we could solve the independent set problem in the following way. If a bag consists of a number of cliques we know there is always one vertex in a clique in the independent set. So even if we have a bag with large cliques, we know that the number of different cases we have to check for each such bag is bounded. 

Tree-clique width is another parameter that attempts to handle dense graphs. While tree-width on dense graphs is unbounded, tree-clique width can be bounded since the width measure is the number of cliques inside a bag in contrast to the number of vertices. Informally, tree-clique width measures how well a graph can be separated into cliques while keeping the number of cliques inside each bag as small as possible. It is easy to realize the benefit of such an approach for problems like vertex cover, independent set, dominating set. A tree decomposition with a bounded number of cliques inside each bag naturally creates a bound on the running times of the aforementioned algorithms.

This paper is organised as follows. In section 2 we give some basic definitions and terminology. In section 3 we present the hardness results that we have obtained for tree-clique width. In particular, we prove that the problem of computing tree-clique width is NP-complete even for constant values. We also prove that there is no constant factor approximation algorithm. In section 4 and 5 we show two dynamic programming approaches while the latter is based on the potential maximal cliques (PMCs) framework.
In section 6 we show how the framework of inclusion-exclusion can be used to obtain faster algorithms by speeding up the clique covering computation.
In section 7 we present efficient time algorithms for tree-clique width on special graphs such as cographs and permutations graphs. We conclude with future work and a short discussion of open problems in section 8.

\section{Preliminaries}

We begin by defining some basic definitions. Let $G = (V,E)$ be a simple, undirected graph without vertices having an edge to themselves. The set $V$ is the set of vertices and $E$ the set of edges, where each edge $(u,v)$ is an edge between vertex $u$ and $v$. A vertex $u$ is \textit{adjacent} to $v$ or a \textit{neighbour} of $v$ if $(u,v) \in E$. The number of vertices in $G$ is denoted as $n = |V|$ and the number of edges as $m=|E|$. The \textit{open neighbourhood} of a vertex $v$ is given by $N(v) = \{u \in V : (u,v) \in E\}$ and the \textit{closed neighbourhood} given by $N[v] = N(v) \cup \{v\}$. The degree of a vertex $v$ is given by $\mathrm{deg}(v) = |N(v)|$. A vertex $v$ is \textit{isolated} if $\mathrm{deg}(v) = 0$. For a set $S \subseteq V$ we have $N(S) = \bigcup_{v \in S}N(v)$ and $N[S] = N(S) \cup S$. The maximum degree of the vertices of a graph $G$ is given by $\Delta(G)$.

A subgraph $G[V]$ is a graph where we take the original graph $G$ and take all vertices $V$ and edges $E$ where for the every edge in the resulting graph, both endpoints are in $V$. 

If we subtract a set from $G$, such as $G - S$, then this implies we remove all the vertices and their edges in the set $S$ from the graph $G$. In case we subtract a single vertex, i.e. $G - v$, it is implied that we subtract the set $\{v\}$ from $G$.

A clique of a graph $G$ is a set $X$ of vertices of $G$ with the property that every pair of distinct vertices in $X$ are adjacent in $G$. A maximal clique of a graph $G$ is a clique $X$ of vertices of $G$, such that there is no clique $Y$ of vertices of $G$ that contains all of $X$ and at least one other vertex. Given a graph $G$, its clique graph, denoted by $clg(G)$, is a new graph such that: every vertex of $clg(G)$ represents a maximal clique of $G$ and two vertices of $clg(G)$ are adjacent when the underlying cliques in $G$ share at least one vertex. 
Furthermore, we denote with $\omega(G)$ of a graph $G$ as the number of vertices in a maximum clique in $G$.

Now we define the minimal separators which play a crucial role for the study of tree-width and subsequently also for tree-clique width. Let $a$ and $b$ be two non adjacent vertices of a graph $G=(V,E)$. A set of vertices $S \subseteq V$ is an $a,b$-separator if $a$ and $b$ are in different connected components of the graph $G \setminus S$. A connected component $C$ of $G \setminus S$ is a full component, associated to $S$, if $N(C)=S$. $S$ is a minimal $a,b$-separator of $G$ if no proper subset of $S$ is an $a,b$-separator. We say that $S$ is a minimal separator of $G$ if there are two vertices $a$ and $b$ such that $S$ is a minimal $a,b$-separator.


A tree decomposition of a graph $G=(V,E)$ is a pair $\mathcal{T}=(T,\{X_t\}_{t\in V(T)} )$, where $T$ is a tree whose every node $t$ is assigned a vertex subset  $X_t \subseteq V(G)$, called a bag such that the following three conditions hold:
\begin{itemize}
    \item $\bigcup_{t\in V(T)}X_t=V(G)$
    \item For every edge $(u,v)$ in $E$, there exists a node $t$ of $T$ such that bag $X_t$ contains both $u$ and $v$.
    \item For every vertex $v$ in $V(G)$, the set $T_u=\{t \in V(T)\ : u \in X_t \}$ i.e., the set of nodes whose corresponding bags contain $u$, induces a connected subtree.  
\end{itemize}

\noindent The width of the tree decomposition $\mathcal{T}=(T,\{X_t\}_{t\in V(T)} )$ equals $\max_{t \in V(T)}|X_t|-1$, that is the maximum size of a bag minus 1. The tree-width of a graph $G$, denoted by $tw(G)$, is the minimum possible width of a tree decomposition of $G$.

Path-width is another parameter very similar to tree-width. In particular tree-width constitutes a generalization of path-width. Both parameters were introduced by Robertson and Seymour \cite{robertson1983graph,robertson1986graph}. Path-width is a problem that has been studied a lot and under different names. Some equivalent definitions are interval graph extension, gate matrix layout, node search number and edge search number. Computing a path decomposition is NP-hard and was shown by Arnborg et al. \cite{arnborg1987complexity}.

Formally, a path decomposition can be defined as follows: a path decomposition of a graph $G$ is a sequence $P=(X_1,X_2,...,X_r)$ of bags, where $X_i \subseteq V(G)$ for each $i\in \{1,2,...,r\}$, such that the following conditions hold:

\begin{itemize}
    \item $\bigcup^r_{i=1} X_i = V(G)$. In other words, every vertex of $G$ is in at least one bag.
    \item For every $uv\in E(G)$, there exists $l\in \{1,2,...,r\}$ such that the bag $X_l$ contains both $u$ and $v$.
    \item For every $u \in V(G)$, if $u\in X_i \cap X_k $ for some $i \leq k$, then $u\in X_j$ also for each $j$ such that $i\leq j \leq k$. In other words, the indices of the bags containing $u$ form an interval in $\{1,2,...,r\}$.
\end{itemize}

The width of a path decomposition $(X_1,X_2,...,X_r)$ is $\max_{1\leq i \leq r} |X_i|-1$. The path-width of a graph $G$, denoted by $pw(G)$, is the minimum possible width of a path decomposition of $G$. We can say that path-width measures how close a graph $G$ can be captured by a path-decomposition. Or in other words, how closely a graph $G$ is related to a path. It is more than evident by looking just at the definitions that path-width is very similar to tree-width. The only exception is in their conditions, namely the one that requires the bags containing a vertex $u$ should form an interval instead of connected subtree. Therefore the intuition that led to the introduction of the tree-width is rather straightforward to comprehend.

We also introduce some terminology regarding parts of a tree decomposition that can be useful during the analysis of the given algorithms. Let $X$ be a node in a tree decomposition $T$. The adhesion of $X$ is the intersection of the bags of $X$ and its parent. If $X$ is a root the adhesion is empty. The margin of $X$ is its bag minus its adhesion. The cone of $X$ is the union of the bags of the descendants  of $X$, including $X$. The component of $X$ is its cone minus its adhesion.

\medskip
We now define the clique cover problem, also known as keyword conflict \cite{kellerman1973determination} problem, covering by cliques, edge clique cover or intersection graph basis \cite{gary1979computers}, as the following decision problem: 

\textbf{Edge clique cover}

\textbf{Input:} An undirected graph $G=(V,E)$ and an integer $k \geq 0$.
    
\textbf{Question:} Is there a set of at most $k$ cliques in $G$ such that each edge $E$ has both its endpoints in at least one of the selected cliques?

The clique cover problem is known to be NP-complete \cite{orlin1977contentment}, and remains NP-complete even if we restrict the problem to special graph classes  \cite{chang2001tree}. Edge clique cover belongs also to FPT and the best known fixed-parameter algorithm is a brute force search on the $2^k$-vertex kernel which runs in double exponential time in terms of $k$ proposed by Gramm et al \cite{gramm2009data}. In addition it is known that the clique cover problem can be solved by a dynamic programming algorithm in time $O(2^{(|E(G)|+|V(G)|)})$. The edge clique cover number is denoted by $ecc(G)$ for a graph $G$.

A similar but not equivalent problem is the vertex clique cover problem. A vertex clique cover is a set of cliques that covers all the vertices of the graph. Formally:

\textbf{Vertex clique cover}

\textbf{Input:} An undirected graph $G=(V,E)$ and an integer $k \geq 0$.
    
\textbf{Question:} Is there a set of at most $k$ cliques in $G$ such that each vertex $v$ is in at least one of the selected cliques?

\noindent If a vertex $v$ appears in two cliques, we can simply remove that vertex from one of the cliques, thus resulting in a vertex cover of equal or smaller size. Therefore, it makes sense to define the output of the vertex clique cover problem as a set of disjoint cliques. The vertex clique cover number, denoted by $vcc(G)$, is defined as the smallest possible number of cliques to cover all the vertices of a graph $G$. Due to the fact that the vertex clique covering could be disjoint, it is worth noting that the vertex clique cover number is less or equal than the edge clique cover number.
\begin{corollary}
For any graph $G$, it holds that $vcc(G) \leq ecc(G)$. 
\end{corollary}

We have now defined all the necessary tools and notions in order to introduce the augmented tree decomposition. 
Below, we give the definition of an augmented tree decomposition of a graph $G=(V,E)$. For such augmented tree decompositions, all properties of tree decompositions hold, but in addition, we augment the bags with coverings of their vertices, and define the width in a different way.

\medskip

\begin{defn} \label{Tree-clique-width-def}
More formally, an augmented tree decomposition is a triple 
$\mathcal{T}=(T,\{X_t\}_{t\in V(T)},\{C_t\}_{t\in V(T)} )$, where $T$ is a tree whose every node $t$ is assigned a vertex subset $X_t \subseteq V(G)$ called a bag, and for each t, $C_t$ is a collection of cliques in G and $ \bigcup_{S \in C_t} S \subseteq X_t$. Also, each edge between vertices in $X_t$ belongs to at least one clique from $C_t$. The following conditions have to hold:

\begin{itemize}

    \item $\bigcup_{t\in V(T)}X_t=V(G)$
    
    \item $\bigcup_{S \in C_t} S = X_t$.
    
    \item For every edge $(u,v)$ in $E$, there exists a node $t$ of $T$ such that bag $X_t$ contains both $u$ and $v$.
    
    \item For every vertex $v$ in $V(G)$, the set $T_u=\{t \in V(T)\ : u \in X_t \}$ i.e., the set of nodes whose corresponding bags contain $u$, induces a connected subtree.
\end{itemize}

For a graph $G$ the width-weight of a bag $t$, is the minimum number of cliques needed to cover every vertex in $X_t$, namely $|C_t|$. The cliques of each bag in $G$ can be considered as given from a clique cover algorithm. The width of an augmented tree decomposition $\mathcal{T}=(T,\{X_t\}_{t\in V(T)},\{C_t\}_{t\in V(t)} )$ is the maximum weight-width of any bag, $\max_{\forall t \in T} |C_t|$. The tree-clique width of a graph $G$, denoted by $tcl(G)$ is the minimum possible width of an augmented tree decomposition of $G$.
\end{defn}


Similarly to path-width we can define the path-clique width. Informally, that means that we would like to have a decomposition of a graph $G$ but this time the indices of the bags containing $u$ form a connected interval instead of connected subtree.

\begin{defn} 
More formally, an augmented path decomposition of a graph $G$ is a triple $\mathcal{P}=(P,\{X_t\}_{t\in V(P)},\{C_t\}_{t\in V(P)})$ of bags, where $P$ is a path whose every node $t$ is assigned a vertex subset $X_t \subseteq V(G)$ called a bag where $X_t \subseteq V(G)$ for each $t\in \{1,2,...,r\}$, $C_t$ is a collection of cliques in G and $ \bigcup_{S \in C_t} S \subseteq X_t$. The following conditions hold: 

\begin{itemize}

    \item $\bigcup_{t\in V(T)}X_t=V(G)$. In other words, every vertex of $G$ is in at least one bag.
    
    \item For every $uv \in E(G)$, there exists $l \in \{1,2,...,r\}$ such that the bag $x_l$ contains both $u$ and $v$.
    
    \item For every $u\in V(G)$, if $u\in X_i \cap x_k$ for some $i \leq k$, then $u\in x_j$ also for each $j$ such that $i \leq j \leq k$. In other words, the indices of the bags containing $u$ form an interval in $\{1,2,...,r\}$
    
\end{itemize}

For a graph $G$ the width-weight of a bag $t$, is the minimum number of cliques needed to cover every vertex in $X_t$, namely $|C_t|$. The width of an augmented path decomposition $\mathcal{P}=(P,\{X_t\}_{t\in V(P)},\{C_t\}_{t\in V(P)} )$ is the maximum width-weight of any bag, $\max_{\forall t \in P} |C_t|$. The path-clique width of a graph $G$, denoted by $pcl(G)$ is the minimum possible width of an augmented path decomposition of $G$.
\end{defn}

A careful parsing of Definition \ref{Tree-clique-width-def} might give rise to the following question: is it clear whether or not the existing cliques are getting disrupted into smaller pieces after applying a tree decomposition? This is a crucial question since the main aim of tree-clique is to create an augmented tree decomposition while minimizing the number of cliques inside every bag. In other words, the cliques of $G$ should be preserved and completely contained in at least one bag $X_t$. The answer to that question is the containment lemma as was firstly observed in \cite{bodlaender1993pathwidth}. Formally:

\begin{lemma}\label{clique-containment-lemma}
Let $\mathcal{T}=(T,\{X_t\}_{t\in V(T)} )$ be a tree decomposition of $G=(V,E)$ and let $W\subseteq V$ be a clique in $G$. Then there exists $t\in V(T)$ with $W \subseteq X_t$.
\end{lemma}

The following corollary is an immediate consequence of Lemma $\ref{clique-containment-lemma}$ applied to an augmented tree decomposition. 

\begin{corollary}
Let $\mathcal{T}=(T,\{X_t\}_{t\in V(T)},\{C_t\}_{t\in V(t)} )$ be an augmented tree decomposition of $G=(V,E)$ and let $W\subseteq V$ be a clique in $G$. Then there exists a $t\in V(T)$ with $W \subseteq X_t$.
\end{corollary}

Throughout this paper for all our algorithms which have an exponential running time we ignore polynomial factors. The notation $O^*(c^n)$ is used to denote the class of functions that are within a polynomial factor of $c^n$. By rounding up $c$ we can omit the star, as a polynomial times an exponential function is growing slower than any exponential function with a sufficiently  large exponent $n$.

\section{Hardness results}

In this section we present hardness results for tree-clique width. In particular, we prove that tree-clique width is NP-complete even for width 3. The immediate consequence of this proof is that tree-clique width is NP-complete in general graphs for any width bigger or equal than 3. Finally, we also prove that there is no constant factor approximation algorithm using similar arguments that are used in the NP-hardness proof.

\subsection{Tree-clique width is NP-complete even for width 3}

In this section we prove that 3-tree-clique width is NP-complete by performing a reduction from 3-coloring. More formally, the 3-tree-clique width problem is the problem of recognising whether a graph $G$ admits a tree decomposition of tree-clique width at most three.

\begin{lemma} \label{Bipartite_sub_con}
\cite{bodlaender1990pathwidth}(Complete bipartite subgraph containment lemma)
\\
Let $(\{X_i|i \in I\},T=(I,F))$ be a tree decomposition of $G=(V,E)$. Let $A,B \subseteq V$, and suppose $\{(v,w)|v \in A, w\in B\}\subseteq E, A \cap B = \emptyset$. Then $\exists i \in I : A \subseteq X_i$ or $B \subseteq X_i$.
\end{lemma}

\begin{lemma} \label{induced_subtree}
\cite{bodlaender1990pathwidth}\\
Let $(\{X_i|i \in I\},T=(I,F))$ be  a tree decomposition of $G=(V,E)$, and let $A,B \subseteq V$ and suppose $\{(v,w)|v\in A, w\in B\}\subseteq E, A\cap B= \emptyset $. Suppose $\exists i \in I: A\subseteq X_i$. Then there exists an induced subtree $T'=(I',F')$ of $T$, such that: 
\begin{enumerate}
    \item $\forall i \in I': A\subseteq X_i$
    \item $B\subseteq \bigcup_{i\in I'}X_i$
    \item $({X'_i|i \in I'}, T'=(I',F'))$ with $X'_i=X_i\cap (A\cup B)$ is a tree decomposition of the subgraph of $G$ induced by $A\cup B$.
\end{enumerate}

\end{lemma}

\begin{theorem} \label{NP-completeness-tree-clique-width}
The 3-tree-clique width problem is NP-complete.
\end{theorem}

We first show that the 3-tree-clique width problem belongs to NP. Given an instance of the problem, we use as a certificate an augmented tree decomposition of a graph $G=(V,E)$ which is a triple 
$\mathcal{T}=(T,\{X_t\}_{t\in V(T)},\{C_t\}_{t\in V(t)} )$. Clearly we can check in polynomial time whether $\mathcal{T}$ is a valid augmented tree decomposition, and whether every bag $X_t$ is indeed covered by the corresponding cliques in $C_t$. In addition we can check in polynomial time that the cardinality of each $C_t$ is at most 3, namely $|C_t| \leq 3, \forall {t\in V(t)}$. Therefore, 3-tree-clique width $\in$ NP.

Now we describe how to construct the graph $H$ which we use in the main proof. Initially we construct the complement of G, namely $G'$, which can be done in polynomial time. Then we create four more disjoint vertices and we draw edges from each of them to every vertex in $G'$. The suggested construction can be seen on Figure \ref{fig:3-tree-clique_width_fig1}. Every vertex of $G'$ is connected with the four vertices on the right of Figure \ref{fig:3-tree-clique_width_fig1}. Let us call this new graph $H$. A characteristic of this new graph $H$, that we are going to exploit, is that it has the same clique covering number as $G'$.

\begin{figure}[H]
    \centering
    \includegraphics[width = 10cm,height=5cm]{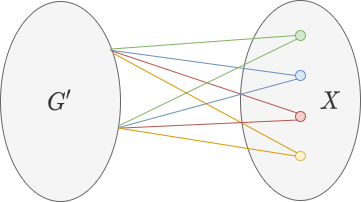}
    \caption{The graph H can be constructed in the following way. First we calculate the complement of $G$, namely $G'$. Then we create 4 disjoint vertices and we draw edges to every vertex in $G'$. The resulting graph $H$ has the same clique covering number as graph $G'$.}
    \label{fig:3-tree-clique_width_fig1}
\end{figure}

We now prove that 3-tree-clique width problem is NP-hard by showing that 3-coloring $\leq_P$ 3-tree-clique width.

\begin{lemma} \label{NP_hardness_proof}
A graph $G=(V,E)$ is 3-colourable if and only if the constructed graph $H$ has a tree decomposition of tree-clique width at most 3.
\end{lemma}
\begin{proof}

Suppose that we have an instance of a 3-coloring on a graph $G$. We notice that since every vertex from $X$ is connected to every vertex in $G'$, this means that the clique covering of $G'$ is the same as in $H$, namely three. This construction can be done in polynomial time. We now show how we can create a tree decomposition of tree-clique width at most 3 in the following way. We create four bags and each one of them contains the graph $G'$ plus one more vertex from $X$, see Figure \ref{fig:3-tree-clique_width_fig2}. Clearly this tree decomposition $\mathcal{T}$ has tree-clique width at most 3 since every $v \in X$ is covered by the three cliques that cover $G'$. Additionally, $\mathcal{T}$ is a valid tree decomposition since $\bigcup_{t\in V(T)}X_t=V(G')$ and for every edge $(u,v)$ from the set $X$ to $G'$, in graph $H$, there exists a bag in $\mathcal{T}$ containing both $u$ and $v$. Finally, it is apparent that such a tree decomposition respects the subtree connectivity requirement.

\begin{figure}[H]
    \centering
    \includegraphics[width = 10cm]{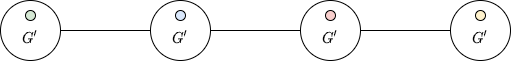}
    \caption{A tree decomposition where every bag contains all the vertices of $G'$ plus one vertex from the four disjoint vertices of the set $X$ in $H$. In this graph the tree-clique width is exactly the same as the clique covering number of $G'$.}
    \label{fig:3-tree-clique_width_fig2}
\end{figure}

Assume now that we have an instance of a tree decomposition $\mathcal{T}$ of a graph as $H$, which has tree-clique width 3. From \cref{Bipartite_sub_con} we know that $\mathcal{T}$ will have a form as Figure \ref{fig:3-tree-clique_width_fig2} suggests. In particular, in $\mathcal{T}$ there must exists an $i\in I:\ A=G'\subseteq X_i$ or $B=X\subseteq X_i$, where $G'$ and $X$ are the left and the right set of vertices of graph $H$, see Figure \ref{fig:3-tree-clique_width_fig1}. Therefore, it is apparent that $A\subseteq X_i$ since the tree-clique width of $\mathcal{T}$ is 3. Otherwise, in the case where $B \subseteq X_i$, the four disjoint vertices of the set $X$ in graph $H$ requires four cliques to be covered. As $A$ is a subset of $X_i$, we know that $A$ is a subset of a set that can be covered by three cliques. Thus $A$ can be covered by three cliques. Thus, the complement of $A$, namely $G$ is 3-colorable.

\end{proof}

This transformation again can be done in polynomial time. Therefore, the 3-tree-clique width problem $\in $ NP-complete.


Since a tree decomposition is a generalization of path decomposiition, we can safely assume that \cref{Bipartite_sub_con} and \cref{induced_subtree} also hold for the case of path decompositions. This means that we can construct the graph $H$ exactly in the same way as we do in the tree-clique width case, see Figure \ref{fig:3-tree-clique_width_fig1}. Thus given a graph $G'$ which is 3-colourable we can always construct a path decomposition as Figure \ref{fig:3-tree-clique_width_fig2} suggests, which has a path-clique width at most 3. As we can observe it is indeed a path decomposition and for every vertex $u\in V(P)$, the indices of the bags containing $u$ form and interval as the third condition of path-clique width definition suggests. Similarly, for the opposite direction, we can prove that given an instance of a path decomposition $\mathcal{P}$ of a graph $H$, which has path-clique width at most 3 we can obtain the graph $G$ which is 3-colourable. Again for this case we can make use of \cref{Bipartite_sub_con} and \cref{induced_subtree} from which we obtain a set $A=G \subseteq X_i$ which is 3-colourable. Therefore, the 3-path-clique width problem $\in$ NP-complete

It is  also worth noting that an immediate corollary of Theorem \ref{NP-completeness-tree-clique-width} is the following.

\begin{corollary}
There is no fixed parameterized algorithm for tree-clique width, unless P=NP.
\end{corollary}

\subsection{The inapproximability of tree-clique width}

In this section we prove that the tree-clique width can not be approximated within a constant factor c, unless of course P=NP. For this purpose we use the 3-colouring problem which is a known NP-complete problem. A rough sketch of the proof idea is as follows. We take an instance of a 3-colouring problem and we transform it to an instance of tree-clique. The transformation is very similar to the graph $H$ of Figure \ref{fig:3-tree-clique_width_fig1} that we used for the NP-completeness proof. Then we show that finding an approximate solution of this new instance of tree-clique width corresponds to finding an approximate solution for the 3-coloring problem.

\begin{theorem}
There is no constant factor approximation algorithm for tree-clique width.
\end{theorem}

Let us now elaborate on the steps of the proof. Suppose we have a graph $G=(V,E)$ of which we want to determine if it is 3-colorable. We then transform $G$ into a new graph $H$ in the following way. Initially we construct the complement of $G$, namely $G'$, which can be done in polynomial time. Then we create $n+1$ vertices and we draw edges from each of them to every vertex in $G'$. The suggested construction can be seen on Figure \ref{fig:3-tree-clique_width_approx}. Every vertex of $G'$ is connected with every vertex of the set $X$.

\begin{figure}[H]
    \centering
    \includegraphics[width = 10cm]{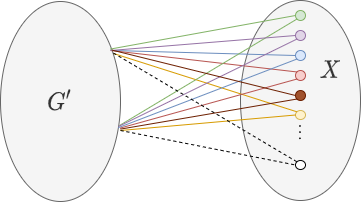}
    \caption{The graph H which has the exact construction process as we do in Figure \ref{fig:3-tree-clique_width_fig1}, but this time the set $X$ contains $n+1$ vertices instead of four. Again the resulting graph $H$ has the same clique covering number as graph $G'$.}
    \label{fig:3-tree-clique_width_approx}
\end{figure}

\begin{lemma} \label{3-colourability-approx_lemma_1}
A graph $G=(V,E)$ is 3-colourable if and only if the constructed graph $H$ has a tree decomposition of tree-clique width at most 3.
\end{lemma}

\begin{proof}

The proof of this Lemma is identical to the proof of \cref{NP_hardness_proof}. The only difference in this case is that the set $X$ does not contain just four vertices but $n+1$. This fact again forces a potential tree decomposition algorithm to include all the vertices of $G'$ in one bag. Therefore, since we know that $G$ is 3-colourable we end up with a tree decomposition of tree-clique width at most 3. For the proof of the other direction we can use again the arguments used in \cref{NP_hardness_proof}.
\end{proof}

\begin{lemma} \label{k-colourability}
A graph $G=(V,E)$ is $k$-colourable if and only if the constructed graph $H$ has a tree decomposition of tree-clique width at most $k$.
\end{lemma}

\begin{proof}
The proof of this lemma follows closely the argumentation used in both \cref{NP_hardness_proof} and \cref{3-colourability-approx_lemma_1}, however we give the full proof for completeness. Given a graph $G$ which is $k$-colourable we can construct a graph $H$, see Figure \ref{fig:3-tree-clique_width_approx}, where the vertex set $X$ contains $n+1$ vertices, where $n > k$. Again using \cref{Bipartite_sub_con} we know that a potential tree decomposition algorithm will be containing all the vertices of $G'$ in one bag which will result in a tree-clique width $k$. We can construct a tree decomposition of tree-clique width $k$ in the same way we do in \cref{NP_hardness_proof}. Namely, we create $n+1$ bags and each one of them contains the graph $G'$ plus one more vertex from $X$, see Figure \ref{fig:tree-clique_width_approx_fig2}. Clearly this tree decomposition $\mathcal{T}$ has tree-clique width at most $k$ since every $v\in X$ is covered by the $k$ cliques that cover $G'$. 

Assume now that we have an instance of a tree decomposition $\mathcal{T}$ of a graph as $H$, which has tree-clique width $k$. From \cref{Bipartite_sub_con} we know that $\mathcal{T}$ will have a form as Figure \ref{fig:tree-clique_width_approx_fig2} suggests. In particular, in $\mathcal{T}$ there must exist an $i\in I:\ A=G'\subseteq X_i$ or $B=X\subseteq X_i$, where $G'$ and $X$ are the left and the right set of vertices of graph $H$, see Figure \ref{fig:3-tree-clique_width_fig1}. Therefore, it is apparent that $A\subseteq X_i$ since the tree-clique width of $\mathcal{T}$ is $k$. Otherwise, in the case where $B \subseteq X_i$, the four disjoint vertices of the set $X$ in graph $H$ requires four cliques to be covered. As $A$ is a subset of $X_i$, we know that $A$ is a subset of a set that can be covered by three cliques. Thus $A$ can be covered by $k$ cliques. Thus, the complement of $A$, namely $G$ is $k$-colorable.

\end{proof}

\begin{figure}[H]
    \centering
    \includegraphics[width = 10cm]{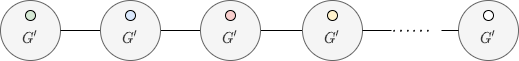}
    \caption{This graph depicts a potential tree decomposition $\mathcal{T}$ of tree-clique width at most $k$ of a graph $H$, where $|X|=n+1$, see Figure \ref{fig:3-tree-clique_width_approx}. In this $\mathcal{T}$ every bag contains all the vertices of $G'$ plus one vertex from the $n+1$ disjoint vertices of the set $X$ in $H$.}
    \label{fig:tree-clique_width_approx_fig2}
\end{figure}

\begin{theorem}\label{UGC_graph_coloring}
If Unique Games Conjecture is true, then there is no constant factor approximation for 3-coloring.
\end{theorem}

Suppose now there is a polynomial time $c$-approximation algorithm for tree-clique width. Now, we can build a polynomial time algorithm, that either determines that a given graph $G$ is not 3-colorable, or decides that $G$ has a $3c$-coloring, as follows. Create the graph H as discussed above with $|X|$ at least $3c+1$. Run the c-approximation algorithm for tree-clique width on $H$.
If this algorithm tells us that the tree-clique width of $H$ is at most $3c$, then by Lemma \ref{k-colourability}, $G$ has a $3c$-coloring. On the other hand, if the algorithm gives us a tree-clique width decomposition of width at least $3c+1$, then the tree-clique width of $H$ is at least $ \lceil (3c+1)/c \rceil =4$, and thus by Lemma \ref{3-colourability-approx_lemma_1}, $G$ is not 3-colorable. That implies that we have found a $c$-approximation algorithm for 3-coloring which contradicts Theorem \ref{UGC_graph_coloring}.

\section{A dynamic programming approach for tree-clique width}

In this section we show how we can use dynamic programming to recognise a tree-clique width decomposition of width $k$. If it exists, we can construct the corresponding decomposition in the same time bound. In particular our algorithm follows closely the algorithm of Arnborg et al. \cite{arnborg1987complexity} for tree-width computation. 

The main idea of the algorithm is the following. For a given graph $G$ there are at most $n^k$ separators of size $k$. Each such separator $S$ creates at most $n$ distinct components upon its removal in $G[V \setminus S]$. This implies that there are only $n^{k+1}$ candidates for pairs of the form: an adhesion between a node $x$ and its parent together with the component of $x$. Each such pair can be obtained by taking a set $S$ of at most $k$ vertices and choosing one connected component of $G[V \setminus S]$. These pairs are then used to form a suitable state space for a dynamic programming procedure that builds up a tree decomposition of width at most $k$ in a bottom up fashion.

Arnborg's et al. \cite{arnborg1987complexity} algorithm was the first dynamic programming algorithm to compute the tree-width of a graph $G$ in polynomial time, given that $k$ is a constant. So, it becomes apparent that such an algorithm is a good starting point to devise an efficient algorithm also for tree-clique width computation. 

\subsection{Algorithm}

As we have already mentioned our algorithm for tree-clique width exploits the same idea from \cite{arnborg1987complexity} but in this case we consider the set of vertices that can be covered by $k$ cliques. Therefore, to do so we have to calculate those sets along with the covering of every set of vertices in $G$ which might not be a separator. The later type of sets will be useful when computing the actual tree decomposition. This type of calculation is done by the \textsc{calculate-cliques} algorithm, which is a preprocessing subroutine. Knowing the covering of every possible subset of vertices in $G$ we can then calculate the actual tree decomposition as Algorithm 1 suggests.

\begin{algorithm}[H] \label{Dynamic-prog-tree-clique-width}

\SetKwData{SetVZero}{MinV0}
\SetKwData{SetVOne}{MinV1}
\SetKwData{SetVTwo}{MinV2}
  \Input{A graph $G = (V,E)$ with $n$ vertices and a parameter $k$}
  \Output{Yes or No}
  \DataStructure{Family of vertex sets that can be covered by $k$ cliques which are separators of $G$. For each such set $s \in S$, there is a set of $l$ connected components of $G$ into which $G$ is separated by removal of $s$. Denoting  $s$ by $C_i$, we denote by $C_i^j$, $1\leq j \leq l$ the subgraphs of $G$ each induced by $s$ and the vertices of the corresponding connected component, with the addition of edges required to make the subgraph induced by $s$ complete. Each such $C_i^j$ has an answer Yes or No (whether it is embeddable in a $k$-tree-clique decomposition or not) determined during the computation.}
  
  \tcc{\textsc{Calculate-Cliques} subroutine computes the superset $S$ which contains all the sets $s$ that can be covered by $k$ cliques in a bottom-up fashion}
  $S=$\textsc{Calculate-Cliques} 
 
  \For{ each set $s \in S$}{
    \If{ $s$ is a separator of $G$}
     {
     Insert $C_i=s$ and the corresponding graphs $C_i^j$ into the data structure
     }
   }
   \textbf{Sort} all graphs $C_i^j$ by increasing size.\\
   \tcc{Tree-clique width calculation}
  \For{each graph $C_i^j$ in increasing order of size $h$}{
        
        \For{$v \in C_i^j$}{
            Examine all $k$-clique separators $C_m$ contained in $C_i \cup \{v\}$\\
            Consider all $C_m^l$ in $(C_i^j-C_i)\cup C_m$ which have tree-clique \\ width $k$
            
            \If{their union, over all l's and all m's contains $C_i^j-C_i$}{
            Set the answer for $C_i^j$ to YES and \textbf{exit-do}
            }
        }

     \If{no answer was set for $C_i^j$}{
        Set the answer for $C_i^j$ to NO}
    \If {$G$ has a separator $C_m$ such that all $C_m^l$ graphs have answer YES}{
        $G$ has a tree-clique decomposition of width $k$\\ 
        \Return YES \\ }
    \If{each separator $C_m$ of $G$ has a $C_m^l$ with answer NO}{
        $G$ has not a tree-clique width decomposition of width $k$. \\
        \Return NO
    }
 }    
  \caption{Dynamic programming for tree-clique width}
\end{algorithm}

\pagebreak

The \textsc{calculate cliques} algorithm is used to calculate the minimum clique covering of a given graph $G$. In order to do so we exploit the fact that the clique cover of a graph $G$ may be seen as a graph coloring of the complement graph of $G$. Therefore, we use Lawler's coloring algorithm \cite{lawler1976note} which finds the chromatic number of each subset of $G'$. The crux of the algorithm as Lawler noticed is that each one of the color classes can be assumed to be a maximal independent set. By using the Moon and Moser bound \cite{moon1965cliques}, we know that there are $O(3^{n/3})$ maximal independent sets in a graph $G$ and we can calculate them in $O(3^{n/3})$ time.

\begin{algorithm}[H]
\SetKwData{SetVZero}{MinV0}

  \Input{A graph $G$ of size $n$ and a parameter $k$}
  \Output{A collection of sets of vertices $\Pi$ that can be covered by $k$ cliques in $G$}
  \Data{X[S] contains the minimum chromatic number of the induced subset $S$ of $G'$. In addition in $X[S]$ we store the actual coloring of the vertices of $G'[S]$, which partition $G$ into $k$ cliques. }
  
  $G'$ = Calculate the complement graph of $G$ 
  
  \For{$S=0$ to $2^n-1$}{
    \For{all $I \in $ MISs of $G'[S]$}{
            $X[S]= \min\{X[S],X[S \setminus I]\}$
            } 
        }
  
  $\Pi$ = $X[ \textrm{All the sets of vertices that admit a $k$ coloring}]$
  \\ 
  \Return $\Pi$

  \caption{\textsc{Calculate-cliques }} \label{Calculate-clique-slow}
\end{algorithm}

For an optimal tree decomposition of tree-clique width at most $k$, we assume the following connectivity condition.  For every node $x$, the component of $x$, the set of vertices appearing in the bags of $x$ and its descendants but not in the parent of $x$, induces a connected subgraph of $G$. If that is not the case we can refine-sanitize the tree decomposition so that the width does not change and the tree decomposition induces a connected subgraph of $G$. Next we give the definition of a sane tree decomposition.

\begin{defn} \label{Sane tree decomposition}
A tree decomposition $t$ of a graph $G$ is called sane if the following conditions are satisfied for every node $x$:

\begin{enumerate}
    \item the margin of $x$ is nonempty
    \item the subgraphs induced in $G$ by the cone of $x$ and by the component of $x$ are connected 
    \item every vertex of the adhesion of $x$ has a neighbor in the component of $x$.
\end{enumerate}

\end{defn}

\begin{lemma} \label{sanitization}
Suppose $t$ is a tree decomposition of a graph $G$. Then
there exists a sane tree decomposition $s$ of $G$ where every bag in $s$ is a subset of some bag in $t$. In particular, if $tw(G) \leq k$, then $G$
admits a sane tree decomposition of width at most $k$.
\end{lemma}
\begin{proof}
We omit the proof which can be found in \cite{bojanczyk2016definability}, Lemma 2.8.
\end{proof}

The correctness of Algorithm \ref{Dynamic-prog-tree-clique-width} is based on the fact that a tree decomposition of tree-clique width $k$, is $k$-decomposable. Informally, a graph $G$ is $k$-decomposable if it can be disconnected by removal of at most $k$ cliques(separator), so that each of the resulting connected components augmented by the connected separator has tree-clique width at most $k$.

\begin{defn}
A graph $G$ is $k$-decomposable if either 1 or 2 holds:
\begin{enumerate}
    \item $G$ has at most $k$ cliques
    \item $G$ has a separator $S$, $ecc(S) \leq k$, such that the components of $G(V-S)$ are $S_1,...,S_n$ and all graphs $G(S_i \cup S) \cup c(S),\ i=1,...,n $ are $k$-decomposable, where $c(S)$ is the induced complete graph $G[S]$. 
\end{enumerate}
\end{defn}

The following lemma reflects that a graph $G$ of tree-clique width at most $k$, is $k$-decomposable.

\begin{figure}[H]
    \centering
    \includegraphics[width = 10cm]{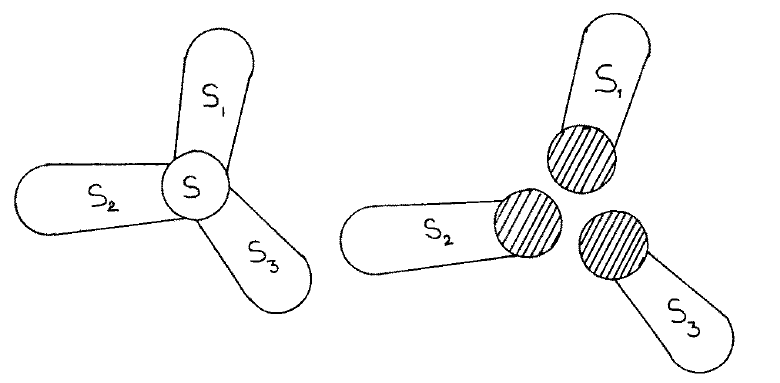}
    \caption{A decomposable graph with its component graphs as it was presented by Arnborg \cite{arnborg1985efficient}.}
    \label{}
\end{figure}

\begin{lemma}
A given graph $G$ of size $n>k$ has a tree-clique width of size at most $k$, if and only if there exists a $k$-clique separator $C_i$ such that all subgraphs $C_i^j$ (as defined in the algorithm) have a tree-clique width of size at most $k$.
\end{lemma}
\begin{proof}
Assume that a given graph $G$ has a tree decomposition $\mathcal{T}$ of tree-clique width $k$ while it does not exist a separator $C_i$ with clique covering less or equal than $k$. This leads us to a contradiction since from the existence of a tree decomposition with $tcl(G)=k$ we know that every bag in $\mathcal{T}$ must have width at most $k$. Subsequently, all the subgraphs $C_i^j$ must also have $tcl(G) \leq k$.

Let us now assume that there exists a $k-$clique separator $C_i$ such all subgraphs $C_i^j$ have a $tcl(C_i^j) \leq k$ and the graph $G$ has a tree-clique width of size at least $k$. However, from the assumption we have made we know that we can construct a tree-decomposition $\mathcal{T}$ of size at most $k$ by simply taking $C_i$ as a base and by attaching on it the $C_i^j$s. If the resulting graph is not sane we can always sanitize it using Lemma \ref{sanitization} and obtain an optimum decomposition.  This will lead in a tree decomposition of tree-clique width at most $k$ contradicting the initial assumption. 
\end{proof}

\begin{theorem} \label{calculate-cliques}
The \textsc{calculate-cliques} algorithm can be implemented to run in $O(2.4423^n)$ time.

\begin{proof}
The \textsc{calculate-cliques} algorithm goes through all the sets of $S$ in increasing cardinality. Then it generates all the maximal independent sets $I$ of $G[S]$. Therefore, we can use the Moon and Moser bound \cite{moon1965cliques} which states that the number of maximal independent sets of a graph is at most $O(3^{n/3})$. We do also know that we can find all these maximal independent sets within a polynomial factor of this bound \cite{tsukiyama1977new}. Therefore we get,

\[O(\sum_{S \subseteq V}|I(G[S])| )= O(\sum_{i=0}^n \binom{n}{i} 3^{i/3}) = O((1+3^{1/3})^n) = O(2.4423) \]

\end{proof}

\end{theorem}

\begin{theorem}
The dynamic programming for tree-clique width algorithm can be implemented to run in time $O(2.4423^n)$.

\begin{proof}
Line 1, the \textsc{calculate-cliques} algorithm, as we have already proved in \cref{calculate-cliques} takes $O(2.4423^n)$ time. Then lines 2 up to 4 take at most $O(2^n)$ time. Sorting all the graphs $C_i^j$ can be done in $O(2^n)$ time using the bucket sort algorithm. The for loop in line 6 examines at most $O(2^n)$ subgraphs and each exit condition for the algorithm can be checked in constant time per subgraph. In Line 7 there are less than $n$ vertices in a subgraph $C_i^j$ and the access to the separator $C_m$ can be made in constant time. Checking the union of all $C_m^l$ is of order of the size of $C_i^j$. Therefore, the overall running time of the algorithm is $O(2.4423^n)$.
\end{proof}

\end{theorem}

It is worth noting that Algorithm \ref{Dynamic-prog-tree-clique-width} recognizes whether a graph $G$ admits a tree decomposition of tree-clique width at most $k$. It is apparent that we can also build the actual tree decomposition in the same time bound. In addition, instead of just recognising for a given $k$, it is easy to see that Algorithm \ref{Dynamic-prog-tree-clique-width} can be modified to compute a tree-clique width decomposition in the same time bound.

\subsection{A Faster Approach}

It is apparent that the calculation of the chromatic number of each subset of the complement graph $G'$ constitutes the absolute bottleneck of the dynamic programming algorithm for tree-clique width. Therefore, in order to get a better running time we should optimize the \textsc{Calculate-cliques} algorithm. It turns out that there is a way to find the chromatic number of each induced subset $S$ of a graph $G$ faster using Byskov's approach \cite{byskov2004enumerating}. 

Algorithm \ref{Faster-calculate-cliques} is a refined version of Algorithm \ref{Calculate-clique-slow}. In particular, the first part of the algorithm, namely lines 2 up to 7, checks 1 and 2 colourability in polynomial time and 3-colourability of all subgraphs of $G$. The 3-colourability of each subset can be calculated by using Eppstein's improved algorithm for 3-coloring \cite{eppstein2001improved} in $O(1.3289^n)$ time. 

Now in the second part, lines 8 up to 11, the algorithm finds all maximal independent sets $I$ of the graph and for each checks 3-colourability of all subgraphs $S$ of $G'[V' \setminus I]$. If they are 3-colourable, then $G[S\cup I]$ is 4-colourable. This will find all maximal 4-colourable subgraphs in time $O(2.4023^n)$. Then in the last part the algorithm needs to consider the maximal independent sets of size at most $|S| / X[S]$ starting from $|S| / X[S]$. Using Byskov's improved bound  \cite{byskov2004enumerating} on the number of maximal independent sets, we get that the third part of the algorithm has a running time $O(2.3814^n)$.

\begin{algorithm}[H] 
\SetKwData{SetVZero}{MinV0}

  \Input{A graph $G$ of size $n$ and a parameter $k$}
  \Output{A collection of sets of vertices $S$ that can be covered by $k$-cliques}
  \Data{X[S] contains the minimum chromatic number of the induced subset $S$ of $G'$. In addition in $X[S]$ we store all the $k$ colourable subsets of $X[S]$ which partition $G$ into $k$-cliques. }
  
  $G'$ = Calculate the complement graph of $G$ 
  
  Let $X$ be an array indexed from $0$ to $2^n-1$ \\
  \For{$S=0$ to $2^n-1$}{
    \uIf{$ \chi(S) \leq 3$}{
        $X[S]= \chi(S)$}
    \Else{
        $X[S]$= $\infty$
    }
  }
  \For{all MISs $I$ in $G$ }{
    \For{all subsets $S$ of $V' \setminus I$}{
        \If{$X[S]=3$}{
            $X[S \cup I] = \min(X[S \cup I],4)$}
        }
  
  }
  \For{ $S=0$ to $2^n-1$}{
    \If{$4 \leq X[S] < \infty$}{
        \For{all MISs of $G'[V' \setminus S]$ of size at most $|S| / X[S] $}{
            $X[S \cup I]= \min(X[S \cup I],X[S]+1) $}
        }
   }
  
   \Return $X$    
  \caption{\textsc{Calculate-cliques }} \label{Faster-calculate-cliques}
\end{algorithm}

\begin{theorem}
Algorithm \ref{Faster-calculate-cliques} can be implemented to run in time $O(2.4023^n)$ and in $O(2^n)$ space. 

\begin{proof}

Calculating the complement graph $G'$ of a graph $G$ can be done in $O(nm)$ time. Lines 3 up to 7 can be calculated in 

\[\sum_{S\subseteq V} O(1.3289^{|S|})= O(\sum_{i=0}^n \binom{n}{i} 1.3289^i ) = O(2.3289^n) \]

In essence the algorithm runs Eppstein's algorithm for 3-coloring \cite{eppstein2001improved} on all subgraphs of $G'$. 

Then the lines from 8 up to 11, using the bound $\lfloor \frac{n}{k} \rfloor^{(\lfloor \frac{n}{k} \rfloor +1)k-n} (\lfloor \frac{n}{k} \rfloor +1)^{n-\lfloor \frac{n}{k} \rfloor k }$ on the number of maximal independent sets of size $k$ from \cite{byskov2004enumerating} and the fact that we can find them within a polynomial factor of this bound, we get that the running time for finding all maximal 4-colourable subgraphs is $O(2.4023^n)$. Then the algorithm in the third part, lines 12 up to 15, follows Eppstein's improved algorithm for graph coloring \cite{eppstein2002small}. This time though, it needs only to consider maximal independent sets of size at most $|S|/4$ and by using Byskov's improved bound on the number of these \cite{byskov2004enumerating} we get a running time of $O(2.3814^n)$. Therefore, Algorithm \ref{Faster-calculate-cliques} has running time of $O(2.4023^n)$ and uses $O(2^n)$ space.

\end{proof}

\end{theorem}

\section{Tree-clique width using PMCs}

In this section we exploit the notion of potential maximal cliques (PMCs) to calculate the tree-clique width of a graph $G$. In particular we translate an existing algorithm for calculating tree-width by Bouchitt\'e and Todinca \cite{bouchitte2001treewidth,fomin2008exact}, to an algorithm that calculates tree-clique width. Definition \ref{PMC-def} gives the formal definition of potential maximal cliques.

\begin{defn} \label{PMC-def}
A set of vertices $\Omega$ of a graph $G$ is called a PMC if there is a minimal triangulation $H$ of $G$ such that $\Omega$ is maximal clique of $H$.
\end{defn}

\subsection{Why PMCs?}

The main intuition that makes PMCs a useful tool for our purpose is the following. Every PMC has at most $n$ minimal separators or it is almost itself a minimal separator. Therefore, very roughly, by examining all the PMCs we can find out the minimal separators of $G$, which we know that they are contained in every PMC $\Omega$ and build a tree decomposition of minimum width.

As we can notice from \Cref{PMC-def}, the framework of PMCs is based on triangulations. However, the following folklore theorem suggests it can be also used to compute tree-width, since tree-width and triangulation is two different ways to model and compute the same thing.

\begin{theorem} \label{tree-width-triangulation-connection}
For any graph $G$, $tw(G) \leq k$ if and only if there is a triangulation $H$ of $G$ such that $\omega(H) \leq k+1$.
\end{theorem}

\noindent In other words, computing the tree-width of $G$ means finding a (minimal) triangulation with the smallest maximum clique size. Intuitively, we can make a similar observation which will let us use the framework of PMCs for tree-clique width computation. Namely, computing the tree-clique width of a graph $G$ means finding a tree-width decomposition with minimum number of cliques inside a bag.  This simple straightforward observation allows us to use the PMCs for the calculation of tree-clique width. All we need to know is the clique covering of each PMC in $G$. 


More formally to justify why PMCs can be used to calculate the tree-clique width of a graph we can use the following arguments. In particular, a tree-clique width decomposition is also a valid tree-width decomposition. That means from Theorem \ref{tree-width-triangulation-connection} that there is also a valid minimal triangulation $H$ of $G$ such that $\omega(H) \leq tw(G)+1$. Now Theorem \ref{non-crossing_minimal_separators} gives us a useful relationship between the minimal separators of a graph and its minimal triangulations. Having such a strong relation between the minimal triangulations of a graphs and its minimal separators there is no doubt why we would like to use the PMC framework for tree-clique width computation.

\begin{theorem} \cite{parra1997characterizations} \label{non-crossing_minimal_separators}
The graph $H$ is a minimal triangulation of the graph $G$ if and only if there is a maximal set of pairwise non-crossing minimal separators $\{S_1,S_2,...,S_p\}$ of $G$ such that $H$ can be obtained from $G$ by completing each $S_i$, $i\in \{1,2,...,p\}$ into a clique.
\end{theorem}

Now Theorem \ref{PMCs-minimal-separators-relation} establish a relation between the PMCs and the minimal separators of a graph. In particular, by Theorem \ref{PMCs-minimal-separators-relation} we know for every potential maximal clique $\Omega$ of $G$, the sets $S_i(\Omega)$ are exactly the minimal separators of $G$ contained in $\Omega$.

\begin{theorem}\label{PMCs-minimal-separators-relation} \cite{bouchitte2001treewidth} 
Let $K \subseteq V$ be a set of vertices of the graph $G=(V,E)$. Let $\mathcal{C}=\{C_1(K),...,C_p(K)\}$ be the set of connected components of $G\setminus K$ and let $\mathcal{S}= \{S_1(K),S_2(K),...,S_p(K)\}$ where $S_i(K), i\in \{1,2,...,p\}$, is the set of those vertices of $K$ which are adjacent to at least one vertex of the component $C_i(K)$. Then $K$ is a potential maximal clique of $G$ if and only if:

\begin{enumerate}
    \item $G \setminus K $ has no full component associated to $K$, and
    \item the graph on the vertex set $K$ obtained from $G[K]$ by completing each $S_i \in \mathcal{S}$ into clique, is a complete graph.
\end{enumerate}
\end{theorem}

\subsection{Algorithm}

Usually the algorithms that are based on the framework of potential maximal cliques consist of two phases. Namely, there is the first phase where the set $\Pi (G)$ of all potential maximal cliques of the graph $G$ is enumerated. Then in the second phase a dynamic programming over the PMCs is performed usually in $O(poly(n) \cdot |\Pi|)$ time. Our algorithm is based on this framework and in particular it follows closely the algorithm from \cite{bouchitte2001treewidth}. However the main difference here is that, since the width of tree-clique width is the number of cliques of any bag of the tree decomposition, we have first to calculate the clique covering of each PMC of $G$. We do this using a subroutine, namely the \textsc{PMCs k-clique covering} algorithm. This algorithm examines every possible subset of $G$ and computes each covering. It also recognises which of these subsets are PMCs, a procedure which can be done in $O(nm)$ time [Theorem 5, \cite{fomin2008exact}]. Then upon the termination of the \textsc{PMCs k-clique covering} algorithm, all the PMCs of $G$ along with their clique covering have been stored. 

This information is then used from $\textsc{Tree-clique using PMCs}$ to compute the tree-clique width using standard dynamic programming techniques. To do so the notion of a block realization is used which is the same idea that Alborg's et al. dynamic programming algorithm also exploited in \cite{arnborg1987complexity}. In particular, let $S$ be a minimal separator of $G$. We note $\mathcal{C}(S)$ to be the set of connected components of $G \setminus S$. A component $C \in \mathcal{C}(S)$ is a full component associated with $S$ if every vertex of $S$ is adjacent to some vertex of $C$. If $C \in \mathcal{C}(S)$, we say that $(S,C)= S\cup C$ is a block associated with $S$. A block $(S,C)$ is called full if $C$ is a full connected component associated with $S$. Now we define by $B=(S,C)$ to be a block of the graph $G$. Then, the graph $R(S,C)=G_S[S\cup C]$, where $G_S$ is the graph obtained from $G$ by completing $S$, is called the realization of the block $B$. 

These realizations are particularly useful since they form a suitable state of subproblems for a dynamic programming algorithm. This can done as follows. Given a separator $S$ and let $C_1,C_2,...,C_p$ be the connected components of $G \setminus S$. We know that a triangulation of $G$ can be obtained by considering the minimal triangulations of $R(S,C_i)$ for any $i, 1 \leq i \leq p$. We can exploit this idea to compute the tree-width of a graph, subsequently the tree-clique width, using these blocks as was shown in \cite{kloks1997treewidth}. Then all we need is to give a characterization of the minimal triangulations of a realization $R(S,C)$ using the potential maximal cliques $\Omega$ with $S\subset \Omega$ and $\Omega \subseteq (S,C)$ and the minimal triangulations of some $R(S,C_i)$ strictly included in $(S,C)$. The next Theorem follows from [Corollary 4.8,\cite{bouchitte2001treewidth}] and shows how to calculate tree-clique width using the PMCs of a graph $G$.  

\begin{theorem}
Let $(S,C)$ be a full block of $G$. Then 
\[tcl(R(S,C)) = \min_{S \subset \Omega \subseteq (S,C)} \max( ecc(\Omega),tcl(R(S_i,C_i))) \]

\noindent where the minimum is taken over all potential maximal cliques $\Omega$ such that $S \subset \Omega \subseteq (S,C) $ and $(S,C_i)$ are the blocks associated to $\Omega$ in $G$ such that $S_i \cup C_i \subset S \cup C $.
\end{theorem}

\begin{algorithm}[H] \label{Tree-clique-width-PMC} \label{Tree-clique-PMC}
\SetKwData{SetVZero}{MinV0}

  \Input{The graph $G$, all its minimal separators}
  \Output{$tcl(G)$}

  \textsc{clique covering} to calculate the clique covering of every subset of $G$ and the PMCs \\           
  Compute all the blocks (S,C) and sort them by the number of vertices\\
  \For{each full block (S,C) taken in increasing order}{
        \uIf{(S,C) is inclusion-minimal}{
        
        $tcl(R(S,C))= ecc(S \cup C)$
        }    
        \Else{
            $tcl(R(S,C))= \infty$
        }
        \For{each p.m.c. $\Omega$ with $S\subset \Omega \subseteq (S,C)$ }{
        
        compute the blocks $(S_i,C_i)$ associated with $\Omega$ s.t. $S_i \cup C_i \subset S\cup C$\\
        $tcl(R(S,C)):= \min(tcl(R(S,C)),\  \max_i(ecc(\Omega),\  tcl(R(S_i,C_i))))$
        }
 }
 let $\Delta^*_G$ be the set of inclusion-minimal separators of $G$ \\
 $tcl(G):=\min_{S\in \Delta^*_G} \ \max_{C \in \mathcal{C(S)}}\ tcl(R(S,C))$
  \caption{\textsc{Tree-clique width using PMCs }}
\end{algorithm}

\begin{algorithm}[H] \label{PMC-clique-covering-SLOW}
\SetKwData{SetVZero}{MinV0}

  \Input{Graph $G$}
  \Output{A collection of sets of vertices $Y$ along with their clique covering. In $Y$ is also stored whether a set is PMC.}
  
    $G'$ = Calculate the complement graph of $G$ 
  
  \For{$S=0$ to $2^n-1$}{
    \For{all MISs of $G'[S]$}{
            \tcc{X[S] $\equiv$ minimum $k$ such that $G'[S]$ is $k$ colourable}
            $X[S]= \min\{X[S],X[S \setminus I]\}$
            }
        $Y[S]$=$X[S]$       
 \\    \If{$G[S]$ is a PMC of $G$ }{
            Mark $S$ in $Y$  \\
        }
        }
\Return  $Y$

  \caption{\textsc{clique covering }}
\end{algorithm}

\begin{theorem}
Algorithm \ref{Tree-clique-width-PMC} can be implemented to run in time $O(2.4423^n)$.

\begin{proof}

To analyze the running time of Algorithm \ref{Tree-clique-PMC} we have first to obtain the running time of Algorithm \ref{PMC-clique-covering-SLOW} which is a subroutine of the first. Algorithm \ref{PMC-clique-covering-SLOW} is another version of Lawler's algorithm \cite{lawler1976note}  for calculating the chromatic number of a graph similar to Algorithm \ref{Calculate-clique-slow}. However this time for each subset $G'[S]$ we also check whether $G[S]$ is a PMC of $G$ and we mark it so we can return it upon termination. This procedure takes $O(nm)$ time \cite{fomin2008exact} so overall the Algorithm \ref{PMC-clique-covering-SLOW} takes $O(2.4423^n)$ time.

The rest of the Algorithm \ref{Tree-clique-PMC} is identical to Fomin et al. \cite{fomin2008exact} and takes $O(n^3 |\Pi_G|)$ time, where $\Pi$ is the number of potential maximal cliques of a graph $G$. As it was shown in \cite{fomin2008exact} a graph $G$ with $n$ vertices has $O(1.7087^n)$ minimal separators and $O(1.8899^n)$ potential maximal cliques. Therefore, Algorithm \ref{Tree-clique-PMC} takes $O(2.4423^n)$ time and can be implemented to run in $O(2^n)$ space since the running time is being dominated by the time needed to compute the clique covering of each subset of $G$.

\end{proof}
\end{theorem}

As the running time of the algorithm is dominated by clique covering algorithm we can obtain a better running time, exactly as we do with Algorithm \ref{Faster-calculate-cliques} by using a better coloring algorithm such as Eppstein's algorithm, used in Algorithm \ref{Faster-calculate-cliques}. As we show in Section 6 we can even get slightly better results using the framework of inclusion-exclusion along with fast matrix multiplication. It is also worth noting that computing the minimal separators $\Delta_G$ of a graph $G$ requires $O(n \cdot 1.7087^n)$ time according to \cite{fomin2008exact}. Therefore, we can compute the minimal separators in Algorithm \ref{Tree-clique-PMC}, instead of requiring them as input, in the same time bound.

\section{Inclusion-exclusion}

Inclusion-exclusion is a very interesting technique to design fast exponential time algorithms. It is used in several occasions where we would like to count combinatorial objects in an indirect way, mainly because the direct way is often harder. Inclusion-exclusion algorithms go through all subsets of a problem, similar to dynamic programming, but in some cases they do not require exponential space. This technique is by no means a new one, but it has come again to the foreground after the recent breakthrough of Bj\"{o}rklund  and Husfeldt \cite{bjorklund2006inclusion} and Koivisto \cite{koivisto20062}, who solved the graph coloring problem in $O(2^n)$ time.  

The graph coloring problem, subsequently the clique covering of the complement, can be seen as a set covering problem. In a minimum set covering problem we are given a universe $\mathcal{U}$ of elements and a collection $\mathcal{S}$ of (non-empty) subsets of $\mathcal{U}$, and the task is to find a subset $\mathcal{S}$ of minimum cardinality covering all elements of $\mathcal{U}$. The only requirement is that $\mathcal{S}$ can be enumerated in $O(2^n)$ time otherwise it is not possible to achieve that complexity. Some examples of problems, other than the chromatic number, that can be solved using this framework are the Domatic number, partition into Hamiltonian subgraphs, partition into forests, partition into perfect matchings, bounded component spanning forest.

More formally, a set cover $S_1,S_2,...,S_k$ is $k$-cover of $(\mathcal{U},\mathcal{S})$ if $S_i \in \mathcal{S} $, $1\leq i \leq k$, and $S_1 \cup S_2 \cup... \cup S_k =  \mathcal{U}$. Similarly, we can define a $k$-partition of $(\mathcal{U},\mathcal{S})$  if $S_i \in \mathcal{S} $, $1\leq i \leq k$, and $S_1 \cup S_2 \cup... \cup S_k =  \mathcal{U}$ and $S_i \cap S_j = \emptyset$ for all $i \neq j$. We can use the notion of $k$-partitions to design an algorithm for graph colouring using the inclusion-exclusion principle. Each partition, or cover, will be an independent set and the goal will be to cover $G$ using the minimum number of those. 

As we have already established from our previous tries, all our algorithms for tree-clique width are dominated by the time needed to calculate the clique covering of each subset. Therefore, in this section we explore possible ways to speed up our clique covering calculations, using the aforementioned ideas, and see whether we can reduce the running time of our approaches down to $O(2^n)$.

\subsection{Finding cliques faster}

We can obtain a slightly faster algorithm for finding the clique covering of each subset of $G$ by using the technique of inclusion-exclusion. By following the approach of Bj\"{o}rklund  and Husfeldt  
\cite{bjorklund2008exact} we can calculate the chromatic number of a graph in time $O(2.3236^n)$ and space $O(2^n n^{O(1)})$, using inclusion-exclusion and fast matrix multiplication \cite{coppersmith1987matrix}.

In particular, let $\mathcal{M}$ denote the family of maximal independent sets of a graph $G=(V,E)$ and let $c_k(G)$ denote the number of ways to cover $G$ with $k$ distinct, possibly overlapping, maximal independent sets. Then we can determine if a graph $G$ has a chromatic number at most $k$ using the following: $\chi(G) = \min\{k: c_k(G)>0\}$.

Using the following Lemma from \cite{bjorklund2008exact} we can calculate the number of ways to cover $G$ with $k$ distinct maximal independent sets, aka colors. Using the same formula on the complement graph of $G$ we can calculate the clique covering of $G$ within the same time bound.

\begin{lemma}\cite{bjorklund2008exact} \label{I-E-fast approach}
For every vertex subset $S \subseteq  V $ let $\alpha(S)$ denote the number of $M \in \mathcal{M}$ that do not intersect $S$. Then

\begin{align}
c_k(G) = \sum_{S \subseteq V} (-1)^{|S|} \binom{\alpha(S)}{k} 
\end{align}

\end{lemma}

We can compute the chromatic number using the above Lemma for $k=1,2,...,n$. In order to compute $c_k(G)$ we do need also to compute $\alpha(S)$, namely all maximal independent sets in $G[V \setminus S]$, for every set $S \subseteq G$. Using the Moon and Moser bound \cite{moon1965cliques} on the number of maximal independent sets in a graph with $r$ vertices which at most $3^{r/3}$ and the fact that they can be listed in polynomial time, we obtain an overall running time of $O(2.4423^n)$. However, as we have already noted we can use exponential space to store every $\alpha(S)$ for every subset $S$ and use fast matrix multiplication to lower further the running time of the algorithm.

\begin{lemma}\cite{bjorklund2008exact}

The chromatic number of a graph can be found in time $O(2.3236^n)$ and space $2^n n^{O(1)}$.

\end{lemma}

As the authors of \cite{bjorklund2008exact} note these bounds depend on the number of maximal independent sets in the input instance. Subsequently, graphs with fewer maximal independent sets have better guarantees. An example of them is the triangle-free graphs with at most $2^{n/2}$ maximal independent sets \cite{hujtera1993number}.

It is obvious that such an approach is suitable for our needs of finding the clique covering of each set $S \subseteq G $. That is because we can obtain the chromatic number of the complement graph $G'$ of $G$ by computing $c_k(G')$. The different colors found in $G'$ constitute a clique covering in $G$. Furthermore, the algorithm keeps track of the maximal independent sets that avoid each set $S \subseteq G$. That is particularly useful for our dynamic programming approaches, Algorithm \ref{Dynamic-prog-tree-clique-width} and \ref{Tree-clique-PMC}, where the algorithms require to know the clique covering of each subset $S$ of $G$. As we will see in the next section by dropping that information, namely $\alpha(S)$ for each set $S$, it can result in faster running times for the calculation of chromatic number of a graph $G$. 

\subsection{Inclusion-exclusion is not constructive}

Since there is a much faster approach to calculate the chromatic number of a graph in $O(2^n)$ time using inclusion-exclusion \cite{bjorklund2006inclusion,koivisto20062}, it is worth considering whether we can use this approach to speed up our clique covering calculations for tree-clique width. 

The main observation that has led to this breakthrough for finding the coloring of a graph was the following. All the previous approaches starting from Lawler \cite{lawler1976note} and then followed by Eppstein \cite{eppstein2002small}, Byskov \cite{byskov2004enumerating} and Bj\"{o}rklund and Husfeldt  
\cite{bjorklund2008exact} have their running times all rely on the Moon-Moser bound \cite{moon1965cliques} on the number of maximal independent subsets of a graph. Namely, for every subset $S$ of $G$ they generate all the maximal independent sets of this set. However, in order to arrive within a polynomial factor of time $2^n$ one has to use no properties of the set of independent sets at all.

In particular, 

\begin{align} \label{I-E formula}
p_k = \sum_{X\subseteq U} (-1)^{|X|} a_k(X)    
\end{align}

\begin{theorem}\cite{bjorklund2006inclusion} \label{I-E-fast pk partitions}
The number of $k$-partitions $p_k$ can be computed in $O(2^n)$ time and space.
\end{theorem}

However, these algorithms are based on evaluating a formula, they return the number of minimum k-coloring without ever constructing the actual coloring. It is possible though to iteratively contract edges and compute an actual coloring $\chi(G)$ within the same time bounds. Algorithm \ref{I-E-clique-covering} suggests a way to calculate the clique covering of each subset of a graph $G$ using the technique of inclusion-exclusion.

\begin{algorithm}[H] \label{I-E-clique-covering}
\SetKwData{SetVZero}{MinV0}

  \Input{Graph $G$}
  \Output{The superset $\Pi$ containing all the sets of $G$ along with their clique covering}
  
   $G'$ = Calculate the complement graph of $G$ 
   \\
  \For{each set $S \subseteq G$ }{
        $k$ = Compute the chromatic number $\chi(G'[S])$ using I-E formula \ref{I-E formula} \\
        \tcc{$k$-coloring computation}
        $\Pi[S]$=\textsc{Coloring Construction($G'[S]$,$k$)} 
   }         
  
\Return  $\Pi$
  \caption{\textsc{I-E clique covering }}
\end{algorithm}

\begin{algorithm}[H] \label{I-E coloring construction}
\SetKwData{SetVZero}{MinV0}

  \Input{Graph $G$, the chromatic number $k$ of $G$}
  \Output{A coloring of $G$}
  
  \While{$G$ is not complete}{
    Select two non adjacent vertices $u$ and $v$ and add an edge between them \\
    
    Compute the chromatic number $\chi(G)$ using I-E formula \ref{I-E formula} \\
    
    \uIf{ $\chi(G) == k$ }{
        \tcc{ $u,v$ take distinct colors}
        Keep the edge $(u,v)$
        
    }
    \Else{
        \tcc{$u,v$ take the same colors}
        Merge $u,v$ into a single vertex and mark them

    }
  
  }
  \Return G with the marked vertices 
  \caption{\textsc{Coloring Construction }}
\end{algorithm}

Algorithm \ref{I-E-clique-covering} works as follows. For each subset $S$ of $G$ computes the chromatic number $k$ using the inclusion-exclusion Formula \ref{I-E formula}. Then using that information it iteratively tries to construct and actual $k$-coloring for this subset $S$ using Algorithm \ref{I-E coloring construction}. The idea for the inner working of Algorithm \ref{I-E coloring construction} is simple and is based on adding a series of edges between non adjacent vertices which either result in the same chromatic number $k$ or higher. In the first case the two vertices receive the same color, so we contract them, and in the latter different and we keep them as they are. Lemma \ref{coloring-construction-correctness} shows the correctness of Algorithm \ref{I-E coloring construction}.

\begin{lemma}\label{coloring-construction-correctness}
Algorithm \ref{I-E coloring construction} correctly constructs a $k$-coloring for a given graph $G$.

\begin{proof}

During the execution time algorithm \ref{I-E coloring construction} maintains the following invariant. For every pair of chosen vertices we know whether they share the same color or they do have distinct colors and graph $G$ is not complete.

Initially no vertices are chosen and $G$ is not complete. If that is not the case the algorithm exits and for a graph with $n$ vertices we need $n$ colors.

Then during execution, the algorithm selects two non adjacent vertices, connects them and recomputes the chromatic number of this modified version of $G$. If $\chi(G)==k$ then we know that $u,v$ take distinct colors otherwise they do share the same color. That is because this edge $(u,v)$ forces a coloring algorithm to create a bigger chromatic number, which is then found by the inclusion-exclusion algorithm. At this stage by continuing on the while loop our loop invariant is being preserved.

Now we examine what happens when the loop terminates. The condition that forces the while loop to terminate is that $G$ is a complete graph. That means that $G$, after a number of vertex contractions and edge connections, has exactly $k$ vertices. Therefore, we need $k$ distinct colors to color the modified graph $G$. We do also know the color of the rest $n-k$ merged vertices since they are marked during execution. Algorithm \ref{I-E coloring construction} correctly constructs a $k$-coloring of a graph $G$.  

\end{proof}
\end{lemma}

\begin{lemma}
Algorithm \ref{I-E-clique-covering} has a running time of $O(3^n)$.
\begin{proof}
The computation of the complement graph of $G$ can be done in $O(nm)$ time. Lines 2 up to 4 are executed $O(2^n)$ times as many as the subsets $S$ of $G$. Line 3 requires $O(2^n)$ time for each subset $S$ of $G$ which requires in total $\sum_{s=0}^n \binom{n}{s}2^s = O(3^n)$ time. Line 4 which is the running time of algorithm \ref{I-E coloring construction}, requires $O((n-1) \cdot 2^n)$ time, for each subset, since we might need at most $n-1$ vertex contractions for the a completely disconnected graph with $n$ vertices. Again since line 4 is executed for each subset $s\in S\subseteq G$ it requires $O(3^n)$ time.  Therefore, the overall running time of algorithm \ref{I-E-clique-covering} is $O(3^n)$ time.
\end{proof}
\end{lemma}

It is apparent, due the fact that we need to compute Formula \ref{I-E formula} for every subset of $G$ and the fact that inclusion-exclusion is not constructive, we can not obtain a faster algorithm for clique covering computation. Therefore, Theorem's \ref{I-E-fast pk partitions} approach is not suitable for our needs since it does not use any information about the maximal independent sets of $G$. Subsequently, for a faster algorithm, using the framework of inclusion-exclusion we can make use of Lemma \ref{I-E-fast approach}.

\section{Special graph classes}
It is rather tempting to research how hard is to compute tree-clique width on special graph classes. Doing so helps to obtain a better understanding of its computational hardness and provides us with insights on how to tackle bigger and more complicated instances of the problem. In this section we provide polynomial time algorithms for tree-clique width computation on cographs and permutation graphs.  

\subsection{Cographs}

An interesting graph class with a lot of applications is the class of cographs. Many intractable problems on general graphs become polynomially time solvable on cographs. Among these well known problems is also tree-width which can be computed in $O(n)$ time as it was shown by Bodlaender and  Rolf H M\"ohring \cite{bodlaender1990pathwidth}. In this work it is also shown that the same algorithm can be used for path-width since for a given cograph $G=(V,E)$: $tw(G)=pw(G)$.

\begin{figure}[H]
    \centering
    \includegraphics[width = 10cm]{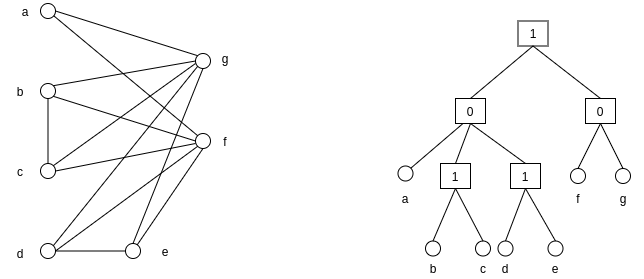}
    \caption{A cograph $G$ and its tree representation $T_G$.}
    \label{fig:cographs1}
\end{figure}

Let us define cographs more formally. Let $G=(V,E),H=(W,F)$ be undirected graphs:
\begin{enumerate}
    \item We denote the disjoint union of $G$ and $H$ by $G \dot\cup H = (V\dot\cup W, E \dot\cup F)$, where $\dot\cup$ is the disjoint union on graphs, and set, respectively.
    \item With $G \times H$ we denote the following type of product of $G$ and $H$: $G \times H= (V \dot\cup W, E \dot\cup F \cup \{(v,w)|v \in V, w \in W\})$ 
    \item The complement of $G$ is denoted by $G'=(V,E')$, with $E'= \{(v,w)|v,w \in V, v \neq w, (v,w) \notin E \}$
\end{enumerate}

A graph $G=(V,E)$ is a cograph, if and only if one of the following conditions hold:

\begin{enumerate}
    \item $|V|=1$
    \item There are cographs $G_1,...,G_k$ and $G= G_1 \dot\cup G_2 \dot\cup ... \dot\cup G_k $
    \item There are cographs $G_1,...,G_k$ and $G= G_1 \times G_2 ... \times G_k$
    
\end{enumerate}

Rule three can be replaced  by the following. There is a cograph $H$ and $G= H'$. There are also other equivalent  characterizations of the class of cographs. Such that every connected subgraph of $G$ has diameter 2 and $G$ does not contain the path graph $P_4$ as an induced subgraph \cite{corneil1981complement}.

Another useful property of each cograph $G=(V,E)$ is that we can associate them with a labelled tree, called the cotree $T_G$. In particular, each vertex of $G$ corresponds to a unique leaf in $T_G$. The internal vertices of $T_G$ have a label $\in \{0,1\}$. We can associate to each vertex of $T_G$ a cotree in the following way. A leaf corresponds to a cotree, consisting of one vertex. The cograph corresponding to a 0-labelled vertex $v$ in $T_G$ is the disjoint union of the cographs, corresponding to the children of $v$ in $T_G$. The cograph corresponding to a 1-labelled vertex $v$ in $T_G$ is the product of the cographs, corresponding to the children of $v$ in $T_G$. For an example of this transformation can be seen on Figure \ref{fig:cographs1}. The cotree is exploited in many algorithms and it is a representation that offers significant complexity gains. Namely, many intractable problems on general graphs are polynomial time solvable on cographs. Some of them are isomorphism, colouring, clique detection, clusters, minimum weight dominating sets, minimum fill-in, hamiltonicity and tree-width \cite{corneil1981complement,corneil1984clustering,bodlaender1990pathwidth}. Therefore, it becomes apparent that we must have a fast algorithm for both recognizing a cograph and producing its unique cotree $T_G$. Corneil et al. \cite{corneil1985linear} gave such an algorithm with a running time $O(n+e)$.

A cotree $T_G$ can easily be transformed to an equivalent binary cotree $T'_G$ such that every internal vertex in $T'_G$ has exactly two children. A binary cotree can be seen on Figure \ref{fig:cotree}. For our algorithms we assume that such a binary tree is given.

\begin{figure}[H]
    \centering
    \includegraphics[height=6cm]{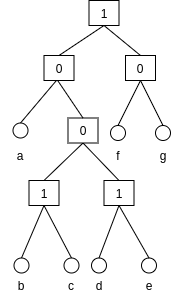}
    \caption{The corresponding binary cotree $T'_G$ of the cotree $T_G$ of Figure \ref{fig:cographs1}. }
    \label{fig:cotree}
\end{figure}

\subsubsection{Algorithm}

Our algorithms are heavily based on the work of Bodlaender and Rolf H M\"ohring \cite{bodlaender1990pathwidth} on tree-width and path-width. These algorithms use a simple recursive top down approach to calculate the product and the disjoint union of every subgraph of a cograph $G$ represented by a subtree of the binary cotree $T_G$. Therefore, we need to know what the result of the tree-clique width of the product and the disjoint union of two graphs is.
The following Lemma addresses this question.

\begin{lemma}
Let $G=(V,E), H=(W,F)$ be graphs. Then:
\begin{enumerate}
    \item $tcl(G \dot\cup H) = \max(tcl(G),tcl(H))$
    \item $tcl(G \times H) = \min(\ \max(ecc(G),tcl(H)),\ \max(ecc(H),tcl(G)))$
\end{enumerate}

\begin{proof}
(1) is trivially true.

(2) First we show that $tcl(G\times H) \leq \max(ecc(H),tcl(G))$. Consider a tree decomposition $\mathcal{T}=(T,\{X_t\}_{t\in V(T)},\{C_t\}_{t\in V(T)} )$ of $G$ with tree-clique width $tcl(G)$. Then $\mathcal{T}=(T,\{X_t \cup W \}_{t\in V(T)},\{C_t\}_{t\in V(T)} )$ is a tree decomposition of $G \times H$ with tree-clique width $\max(ecc(H),tcl(G))$ because $G$ and $H$ are connected. So the $tcl(G\times H) \leq \max(ecc(H),tcl(G))$. Similarly, we can show that $tcl(G \times H)\leq \max(ecc(G),tcl(H))$.

Next we show that $tcl(G \times H) \ge \min( \max(ecc(G),tcl(H)), \max(ecc(H), tcl(G)))$. Consider a tree decomposition $\mathcal{T}=(T,\{X_t\}_{t\in V(T)},\{C_t\}_{t\in V(T)} )$ of $G \times H $. Since every vertex of $G$ is connected to every vertex of $H$ we know from \cref{Bipartite_sub_con} that $\exists t : V \subseteq X_t$ or $\exists t : W \subseteq X_t$. Suppose $\exists t : V \subseteq X_t$. Let $\mathcal{T'}=(T',\{X_t\}_{t\in V(T')},\{C_t\}_{t\in V(T')} )$ be a subtree of $\mathcal{T}$ such that $\forall t \in V(T'): V \subseteq X_t$, $W\subseteq \bigcup_{t \in V(T')} X_t$ and $(T',\{X_t\}_{t\in V(T')},\{C_t\}_{t\in V(T')} )$ is a tree decomposition of $G \times H$. We know that $\mathcal{T'}$ exists by \cref{induced_subtree}. We can observe that $\mathcal{T'}=(T',\{X_t \cap W \}_{t\in V('T)},\{C_t\}_{t\in V(T')} )$ is a tree decomposition of $H$, so $\exists t \in V(T'): ecc(X_t \cap W) \ge tcl(H) $ which means that $\exists t\in V(T'): ecc(X_t) \ge max(ecc(G),tcl(H))$. So the tree-clique width of the tree decomposition $\mathcal{T}=(T,\{X_t\}_{t\in V(T)},\{C_t\}_{t\in V(T)} )$ is at least $max(ecc(G),tcl(H))$.

Similarly, we can show that if $\exists t: W \subseteq X_t$ then the tree-clique width of $\mathcal{T}=(T,\{X_t\}_{t\in V(T)},\{C_t\}_{t\in V(T)} )$ is at least $max(ecc(H),tcl(G))$. Therefore, we can conclude that $tcl(G \times H) = \min(\ \max(ecc(G),tcl(H)),\ \max(ecc(H),tcl(G)))$.
\end{proof}
\end{lemma}

Now that we know how to calculate the tree-clique width of the product and the disjoint union of two graphs we are ready to describe a linear time algorithm for tree-clique width. The $\textsc{Compute-ecc}$ algorithm computes the edge clique cover for each vertex in a binary cotree. Namely, for a vertex $v$ of a binary cotree $T_G$, $\textsc{Compute-ecc}$ computes the minimum clique cover of all existing vertices on the subtree rooted at $v$. $\textsc{Compute-tcl}$ algorithm is another recursive function which computes for every vertex $v$ the tree-clique width of the graph induced by the existing vertices on the subtree of $T_G$ rooted at $v$.

In order to compute the tree-clique width of a cograph $G$, we first call the $\textsc{Compute-ecc}$ given as input the root $r$ of the corresponding cotree $T_G$. We can then call $\textsc{Compute-tcl}$ algorithm, again using $r$ as input, to calculate the tree-clique width of $G$. Only a constant number of operations per vertex is performed therefore, the overall cost is $O(n)$.

\medskip
\begin{algorithm}[H]\label{cographs-comput-ecc}
\SetKwData{SetVZero}{MinV0}

  \Input{A node $v$ of the binary cotree $T_G$}
  \Output{The $ecc$ of the nodes rooted at $v$}

  \uIf {$v$ is a leaf of $T_G$}{
        $ecc(v)=1$ 
   }
  \Else{
      \textsc{compute-ecc}(left child of $v$)\\
      \textsc{compute-ecc}(right child of $v$)\\
  
      \uIf {$label(v)==0$}{
            $ecc(v)=ecc$(left child of $v$)$+ecc$(right child of $v$) 
       }
      \Else{
           $ecc(v)=\max$($ecc$(left child of $v$),\ $ecc$(right child of $v$) )
      
      }
  }
  
  \caption{\textsc{Compute-ecc }}
\end{algorithm}

\begin{algorithm}[H] \label{cographs-comput}
\SetKwData{SetVZero}{MinV0}

  \Input{A node $v$ of the cotree $T_G$, $ecc(v)$}
  \Output{The $tcl$ of the nodes rooted at $v$}
  
  \uIf {$v$ is a leaf of $T_G$}{
    
    $tcl(v)=1$  
  }
  \Else{ 
    \textsc{compute-tcl}(left child of $v$)\\
    \textsc{compute-tcl}(right child of $v$)\\
    
    \uIf{$label(v)==0$}{
        $tcl(v)=\max(tcl(\text{left child of}\ v),tcl(\text{right child of}\ v))$
    }
    \Else{
        
        $tcl(v)=\min($
        \parbox[t]{.6\linewidth}{ $\max(ecc(\text{left child of}\ v),tcl(\text{right child of}\ v))$,\\ $\max(tcl(\text{left child of}\ v),ecc(\text{right child of }v)))$ }
    }
  }
  
  \caption{\textsc{Compute-Tcl }}
\end{algorithm}

\begin{lemma}
The algorithm \textsc{Compute-ecc} correctly computes edge the clique covering of a cograph $G=(V,E)$ given its binary cotree representation $T_G$.
\end{lemma}
\begin{proof}
The algorithm uses a simple top down recursive technique to pass through every vertex $v$ of $T_G$. Every vertex $v$ is labelled with its edge clique cover, $ecc(v)$. Namely, the number of cliques needed to cover $v$ along with every vertex at the subtree rooted at $v$. For this reason every leaf of $T_G$ is labelled with $ecc(v)=1$. 

Now let's examine what happens when a vertex with label zero, $label(v)==0$, is handled. Assume an internal vertex $v$ with $label(v)==0$. It is apparent that since the graphs rooted at $v$ are disconnected we need to consider the sum of both its children. Assume now an another internal vertex $v$ with $label(v)==1$. The subtrees rooted at $v$ are connected therefore their edge clique covering can be seen as the chromatic number of their complements. The complement of the product of two graphs results in the disjoint union of two graphs. Therefore, their chromatic number is just the maximum of two disjoint graphs.
\end{proof}

\begin{theorem}

The tree-clique width of a cograph given its corresponding binary cotree, can be computed in $O(n)$ time.

\begin{proof}

Algorithm \ref{cographs-comput} and \ref{cographs-comput-ecc} both require as input a binary cotree. This binary tree is full, by construction, meaning that each node is either a leaf or possesses exactly two child nodes. Therefore, since the number of leaves in a non-empty full binary tree is one more than the number of internal nodes the running time of both algorithms is $O(n)$.

\end{proof}

\end{theorem}

\subsection{Permutation graphs}

Another class of graphs that is worth considering is the permutation graphs (PGs). They belong to the class of perfect graphs which are the graphs in which the chromatic number of every induced subgraph equals the order of the largest clique of that subgraph (clique number). They were first introduced by Even et al. \cite{even1972permutation,pnueli1971transitive}. More formally, if $\pi$ is a permutation of the numbers $1,2,...,n$ we can construct a graph $G[\pi]=(V,E)$ with vertex set $V=\{1,2,..,n\}$ and edge set $E:$ 
\[(i,j)\in E \iff (i-j)(\pi_i^{-1}- \pi_j^{-1})< 0\]
An undirected graph is a permutation graph if there is a permutation $\pi$ such that $G \cong G[\pi]$. The graph $G[\pi]$ is also called the inversion graph of $\pi$.

\begin{figure}[H]
    \centering
    \includegraphics[scale=0.5]{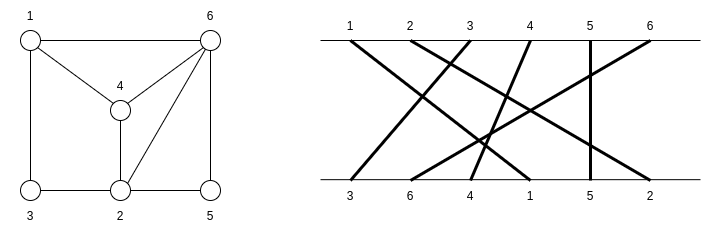}
    \caption{A permutation graph on the left and its corresponding matching diagram.}
    \label{fig:permutation-graphs}
\end{figure}

Permutation graphs may also be defined, geometrically as intersection graphs which is illustrated by a matching diagram.

\begin{defn}
Let $\pi$ be a permutation of $1,2,...,n$, the matching diagram can be obtained as follows. Write the numbers $1,...,n$ horizontally from left to right. Below, write the numbers $\pi_1,...,\pi_n$ also horizontally from left to right. We then draw straight line segments joining the two 1's, the two 2's etc.
\end{defn}

\noindent On Figure \ref{fig:permutation-graphs} we give an example of permutation graph along with its matching diagram. During this paper we assume that the permutation is given and the permutation graph is identified with the inversion graph.

PGs have numerous applications and their main attraction point is the fact that many intractable problems, on general graphs, can be solved efficiently. Among these problems some of the well known ones are the tree-width \cite{meister2010treewidth,bodlaender1995treewidth} maximum clique \cite{mcconnell1999modular,mohring1985algorithmic,supowit1985decomposing}, chromatic number \cite{mcconnell1999modular,mohring1985algorithmic,supowit1985decomposing}, clique cover \cite{mcconnell1999modular,mohring1985algorithmic,supowit1985decomposing}, independent set \cite{mcconnell1999modular,mohring1985algorithmic,supowit1985decomposing}, dominating set \cite{chao2000optimal}, Hamiltonian cycle \cite{deogun1994polynomial}, and graph isomorphism \cite{colbourn1981testing}. The path-width of a graph can also be solved efficiently using any of the approaches introduced in \cite{meister2010treewidth,bodlaender1995treewidth} since Bodlaender et al. \cite{bodlaender1995treewidth} showed that on PGs the tree-width of a graph $G$ is equal to the path-width of $G$. In addition, the all pairs shortest path problem can be solved in faster time than on general graphs \cite{bazzaro2009localized,mondal2003optimal}. 

The PGs can be recognised in polynomial time. The first result was from Pnueli et al. \cite{pnueli1971transitive} which is an $O(n^3)$ time algorithm for recognising permutation graphs. Later, Spinrad \cite{spinrad1985comparability} improved this result by designing an $O(n^2)$ time algorithm. The current fastest algorithm for recognising PGs is a $O(n+m)$ time algorithm from McConnell et al. \cite{mcconnell1999modular}. A survey of properties of permutation graphs can be found in \cite{golumbic2004algorithmic}. Some of them are that every induced subgraph of a permutation graph is a permutation graph and the complement of a permutation graph is a permutation graph. As we will show later our algorithms exploit these properties in order to calculate efficiently the minimum clique covering of each bag of a tree decomposition.

\subsubsection{Scanlines}
Our algorithm for tree-clique width is based on the work of Bodlaender et al. \cite{bodlaender1995treewidth} for tree-width. Namely, in this work it is shown that every minimal separator in a permutation graph can be obtained by using a scanline, which is a line segment on a matching diagram, see Definition \ref{scanline}.

\begin{defn} \label{scanline}
A scanline in the matching diagram is any line segment with one end vertex on each horizontal line. A scanline $s$ is between two line segments $x$ and $y$ if the top point of $s$ is in the open interval bordered by the top points of $x$ and $y$ and the bottom point of $s$ is in the open interval bordered by the bottom points of $x$ and $y$.
\end{defn}

Let us now consider a scanline $s$ which is between two line segments $x$ and $y$. We assume that the intersection of each pair of the three line segments is empty. Therefore, the line segments in the diagram corresponding to $x$ and $y$ do not cross in the diagram. We say that two line segments cross if they have a nonempty intersection. Similarly, we say that a clique $y$ crosses a line segment $s$ if the line segments that $y$ consists of, cross $s$. If we take out all the cliques that cross the scanline $s$, this clearly corresponds to an $x,y-$separator in the graph. The next lemma shows that we can find all minimal separators of a permutation graph in this way.

\begin{lemma} \label{x,y-separator}
Let $G$ be a permutation graph, and let $x$ and $y$ be nonadjacent vertices in $G$. Every minimal $x,y-$separator consists of all cliques crossing a scanline which lies between the line segments of $x$ and $y$.
\end{lemma}

The proof of \ref{x,y-separator} can be found in \cite{bodlaender1995treewidth}. The proof of Lemma \ref{x,y-separator} is identical to the proof of [Lemma 3.1 \cite{bodlaender1995treewidth}] but instead of counting vertices we count the minimum number of cliques crossing $s$.

\begin{defn} \label{scanline_equivalence}
Two scanlines $s_1$ and $s_2$ are equivalent, $s_1 \equiv s_2$, if they have the same position in the diagram relative to every line segment; i.e. the set of line segments with the top (or bottom) end point to the left of the top (or bottom) end point of the scanline is the same for $s_1$ and $s_2$.
\end{defn}

\begin{defn}
A scanline $s$ is $k$-small if it crosses with at most $k$ cliques.
\end{defn}

\begin{lemma}\label{pairwise-scanlines}

There are $O(n^2)$ pairwise non-equivalent $k$-small scanlines.

\begin{proof}

Suppose $\pi=[n,...,2,1]$ is a permutation of the numbers $1,2,...,n$. Clearly, the graph $G[\pi]$ is a clique consisting of $n$ vertices. For any matching diagram we can observe that there are at most $O(n^2)$ minimal separators in a permutation graph with $n$ vertices. That can be easily seen by considering that there are $n$ points at the top and $n$ at the bottom. By trying all their combinations will result in $O(n^2)$ pairwise non-equivalent scanlines. Now observe that since the graph $G[\pi]$ is a clique consisting of $n$ vertices, then any of the $O(n^2)$ possible scanlines will be 1-small. This proves the Lemma.
\end{proof}

\end{lemma}
 
\subsubsection{Candidate components}

We now consider the candidate components which informally, is the number of cliques between two scanlines. In essence, the candidate components are the potential bags of the tree decomposition that our algorithm will have to consider. More formally,

\begin{defn}
Let $s_1$ and $s_2$ be two scanlines of which the intersection is either empty or one of the end points of $s_1$ and $s_2$. A candidate component $C=\mathcal{C}(s_1,s_2)$ is a subgraph of $G$ induced by the following sets of lines:

\begin{itemize}
    \item All lines that are between the scanlines (in case the scanlines have a common end point, this set is empty)
    \item All lines crossing at least one of the scanlines.
\end{itemize}
\end{defn}

The candidate components is the crux of this algorithm since they are the area between two scanlines. Therefore, based on \cite{bodlaender1995treewidth}, we have to redefine what is $k$-feasible candidate component to meet our needs for tree-clique computation.

\begin{defn}
A candidate component $C=\mathcal{C}(s_1,s_2)$ is $k$- feasible, where $k$ an integer, if there is a tree decomposition $\mathcal{T}$ of $C$ such that $tcl(C) \leq k$.
\end{defn}

The following theorem is crucial as it constitutes the correctness of our algorithm. Again for each candidate component we keep the minimum number of cliques, needed to cover every vertex in the given $\mathcal{C}$, instead of just the number vertices it contains.

\begin{theorem} \cite{bodlaender1995treewidth} \label{Scanline_theorem}
Let $C=\mathcal{C}(s_1,s_2)$ be a candidate component with $s_2$ to the right of $s_1$. Then $C$ is $k$-feasible if and only if there exists a sequence of scanlines $s_1=t_1,t_2,...,t_r=s_2$ such that the following conditions are satisfied:
\begin{itemize}
    \item $t_i$ and $t_{i+1}$ have an endpoint in common for $i=1,...,r-1$ and the other end point of $t_{i+1}$ lies to the right of the other end point of $t_i$
    \item Each $\mathcal{C}(t_i,t_{i+1})$ has at most $k$ cliques $(i=1,...,r-1)$.
\end{itemize}
\end{theorem}

\subsubsection{Algorithm}

We are now ready to show how to compute the tree-clique width of a permutation graph. The algorithm works accordingly to \cite{bodlaender1995treewidth} and checks whether the tree-clique width of a permutation graph does not exceed a given integer $k$. We first show how to compute the scanline graph.

We use a directed acyclic graph $W_k(G)$ as follows. The vertices of the graph are the pairwise non-equivalent $k$-small scanlines. Direct an arc from scanline $s$ to $t$ if the following conditions hold:

\begin{itemize}
    \item The intersection of $s$ and $t$ is one end point of $s$ and $t$, and the other end point of $t$ is to the right of the other end point of $s$, and
    \item the candidate component $\mathcal{C}$ consists of at most $k$ cliques.
\end{itemize}

\noindent We call this structure the scanline graph.

\begin{lemma} 
Let $s_L$ be the scanline which lies totally to the left of all line segments and let $s_R$ be the scanline which lies totally to the right of all line segments. Then $G$ has tree-clique width at most $k$ if and only if there is a directed path in the scanline graph from $s_L$ to $s_R$.

\begin{proof}
The proof of this Lemma follows directly from Theorem \ref{Scanline_theorem} since the tree-clique width of graph $G$ equals $\mathcal{C}(s_L,s_R)$.
\end{proof}

\end{lemma}

We are now ready to describe the algorithm which determines if the tree-clique width of a graph $G$ is at most $k$.

\begin{algorithm}[H]
\SetKwData{SetVZero}{MinV0}

  \Input{A permutation graph $G$ along with its permutation $\pi$ and an integer $k$}
  \Output{If $tcl(G) \leq k$ then Yes, otherwise No}

  Make a maximal list $L$ of pairwise $non-equivalent$ $k$-small scanlines \\
  Compute the minimum clique covering of each item in $L$\\
  Construct the acyclic digraph $W_k(G)$ \\
  If there exists a path in $W_k(G)$ from $s_L$ to $s_R$, then $tcl(G) \leq k$. Otherwise, $tcl(G) \ge k$. 
  
  \caption{\textsc{TCL of permutation graph}}
\end{algorithm}

\begin{theorem}
The tree-clique width of a permutation graph can be computed in $O(n^3 \log n )$ time.
\begin{proof}
Clearly computing a maximal list $L$ of pairwise non-equivalent $k$-small scanlines can be done in $O(n^2)$ time according to \cref{pairwise-scanlines}. The minimum clique covering of each scanline in $L$ can be computed in $O(n \log n)$ time using the algorithm of Supowit \cite{supowit1985decomposing} for permutation graphs. This can be done firstly by computing the complement $G'$ of the permutation graph $G$, which is still a permutation graph, and then computing the coloring of $G'$. Using that we obtain the minimum clique covering of each scanline in $G$.  

The scanline graph $W_k(G)$ can be constructed in $O(n^2 k)$ time since there are $O(n^2)$ many scanlines and for each $k$-small scanline we can compute its adjacency list in $O(k)$ time. Computing a path from $s_L$ to $s_R$, if it exists can also take $O(n^2 k)$. Therefore, we get an overall running time of $O(n^3 \log n)$. 
\end{proof}
\end{theorem}

\section{Conclusion}

The main aim of this research was a theoretical study of tree-clique width. One could see this new width measure as tree-width where instead of measuring the number of vertices inside the bags, we measure the number of cliques. A straightforward justification of why we would need such a measure is the following. Most real world graphs are generally sparse, like social network graphs \cite{adcock2016treeSocial}, and they contain a limited number of very dense areas, sometimes also known as hubs. Using tree-clique width as a width measure we can obtain bounds and solve efficient several problems such as the  independent set, the chromatic number, vertex cover etc. The reason for this is that we do have the information about the number of existing cliques in each bag therefore, we know the number of checks we have to perform.

During this research we have provided an NP hardness proof and in particular, we have proved that tree-clique width is NP-complete even for width values of bigger or equal to three. That came as no surprise since finding those set of vertices that can be covered by $3$-cliques is equivalent to finding if the complement graph admits a 3-coloring, which is an already known NP-complete problem. Using similar argument we have also shown that there is no constant factor approximation algorithm for tree-clique width. The immediate consequence of this result is that there is no fixed parameterized algorithm for tree-clique width. These hardness results paved the way for the algorithmic research that followed, in a sense, they heavily reduced our expectations for efficient algorithms.

Algorithmic wise we have constructed two dynamic programming algorithms. The first one is based on Alborg's et al algorithm \cite{arnborg1987complexity} for tree-width and the second one on the well known framework of potential maximal cliques (PMCs). It is apparent that the bottleneck of these approaches is the clique covering computation. To tackle this burden we have used several approaches starting from Lawyer's algorithm for coloring \cite{lawler1976note} and the inclusion-exclusion framework. We have also provided two efficient algorithms for the permutation graphs and the cographs.


Informally speaking the problem of tree-clique width can be seen as two NP-complete problems stacked on top of each other. Namely, the clique covering problem and the problem of building a tree decomposition. Therefore, it is highly unlikely that a tractable algorithm exists for tree-clique width, unless P=NP. In addition, finding a faster algorithm than $O(2^n)$ would require a significant breakthrough since that would also result in a faster algorithm for finding the chromatic number of a graph. 

The results of this research project proved that we have to deal with a really hard, computationally speaking, parameter. However, that does not mean we should withdraw our attention from tree-clique width. Parameters like this can be used in conjunction with others such as tree-width where more efficient algorithms already exist. An example of that, would be to calculate the tree-clique width of a graph only in specific dense areas and use tree-width in the rest more sparse parts of the graph. In the following section, we provide more research ideas that we think they can be used to shed more light on tree-clique width. Therefore, we conclude that tree-clique width is an interesting parameter but further research is required in order to obtain better results and further understand its true potential.

\subsection{Further research}

In this section we propose a series of possible research directions that could be followed for further understanding of tree-clique width. In particular, they can be summarised as it follows:

\begin{itemize}
    
    \item Consider the cliques given from a clique covering algorithm. What is the best algorithm we can design for tree-clique width? Under this condition does the problem of calculating tree-clique width boils down to calculating the tree-width of a clique graph?

    \item Computation of tree-clique width in more special graphs. Some examples of graphs where we it might be possible to get polynomial time algorithms could be the following ones: chordal graphs, interval graphs, planar graphs, bipartite graphs and split graphs. In particular, by observing how fast we can compute tree-clique width on those special graphs in contrast to other parameters can reveal useful properties of the parameter itself. 
    
    \item In addition, we have proved that the tree-clique width is NP-complete even for values greater or equal to three. Therefore, can we calculate the tree-clique width of two in polynomial time?
    
    \item How strong is tree-clique width related to other width parameters, and what is its relationship with clique-width and tree-width? Can we find any bounds of the form $tcl(G) \leq w(G) $ or $w(G) \leq tcl(G) $ for some existing parameter $w$? In addition, what can we find out from the comparison with clique width?
    
    \item As we have proved in Section 3 there is no constant factor approximation algorithm (APX) for tree-clique width. However, it would be interesting enough to answer the following question. Could we obtain an FPTAS, a PTAS or an $f(n)$-APX approximation algorithm for tree-clique width? If the answer is positive, then what is the best factor we can obtain?
    
    \item Another useful research question is what is the relationship of tree-clique width and path-clique width related to the one of tree-width and path-width. For example the path-width and tree-width of a cograph is the same and can be computed in $O(n)$ time \cite{bodlaender1990pathwidth}. Does the same hold for tree-clique width and path-clique width? 
    
    \item A deeper analysis to better understand tree-clique width would require thinking of the theory of minimal separators from the ground up. Namely, the minimal separators are defined by the number of vertices they contain. What algorithmic gains can be obtain by measuring the minimal separators by the minimum number of cliques that cover it instead of the number of vertices? By doing so we alter the meaning of minimality of an a,b-separator. In particular, we propose the following definition of an a,b-separator. A set of vertices $S \subseteq V$ is an $a,b$-separator if $a$ and $b$ are in different connected components of the graph $G \setminus S$. S is a minimal a,b-separator if no proper subset of S with less or equal clique covering separates a,b. To our knowledge to date, there has been no research towards this direction. 
    
    \item An experimental study would be also of great importance. In particular, such a study could reveal graph classes or specific graph structures where tree-clique width can be calculated faster than in other cases. It is also worth mentioning that an experimental study could help in developing heuristics.
    
\end{itemize}

\section*{Acknowledgements}

The author thanks Hans Bodlaender and Jelle Oostveen for discussions and suggestions that helped to improve this paper.

\bibliography{ref}

\begin{thebibliography}{10}

\bibitem{adcock2016treeSocial}
Aaron~B Adcock, Blair~D Sullivan, and Michael~W Mahoney.
\newblock Tree decompositions and social graphs.
\newblock {\em Internet Mathematics}, 12(5):315--361, 2016.

\bibitem{arnborg1985efficient}
Stefan Arnborg.
\newblock Efficient algorithms for combinatorial problems on graphs with
  bounded decomposability—a survey.
\newblock {\em BIT Numerical Mathematics}, 25(1):1--23, 1985.

\bibitem{arnborg1987complexity}
Stefan Arnborg, Derek~G Corneil, and Andrzej Proskurowski.
\newblock Complexity of finding embeddings in a $k$-tree.
\newblock {\em SIAM Journal on Algebraic Discrete Methods}, 8(2):277--284,
  1987.

\bibitem{bazzaro2009localized}
Fabrice Bazzaro and Cyril Gavoille.
\newblock Localized and compact data-structure for comparability graphs.
\newblock {\em Discrete Mathematics}, 309(11):3465--3484, 2009.

\bibitem{bjorklund2006inclusion}
Andreas Bjorklund and Thore Husfeldt.
\newblock Inclusion--exclusion algorithms for counting set partitions.
\newblock In {\em 2006 47th Annual IEEE Symposium on Foundations of Computer
  Science (FOCS'06)}, pages 575--582. IEEE, 2006.

\bibitem{bjorklund2008exact}
Andreas Bj{\"o}rklund and Thore Husfeldt.
\newblock Exact algorithms for exact satisfiability and number of perfect
  matchings.
\newblock {\em Algorithmica}, 52(2):226--249, 2008.

\bibitem{bodlaender1995treewidth}
Hans~L Bodlaender, Ton Kloks, and Dieter Kratsch.
\newblock Treewidth and pathwidth of permutation graphs.
\newblock {\em SIAM Journal on Discrete Mathematics}, 8(4):606--616, 1995.

\bibitem{bodlaender1990pathwidth}
Hans~L Bodlaender and Rolf~H M{\"o}hring.
\newblock The pathwidth and treewidth of cographs.
\newblock In {\em Scandinavian Workshop on Algorithm Theory}, pages 301--309.
  Springer, 1990.

\bibitem{bodlaender1993pathwidth}
Hans~L Bodlaender and Rolf~H M{\"o}hring.
\newblock The pathwidth and treewidth of cographs.
\newblock {\em SIAM Journal on Discrete Mathematics}, 6(2):181--188, 1993.

\bibitem{bojanczyk2016definability}
Miko{\l}aj Boja{\'n}czyk and Micha{\l} Pilipczuk.
\newblock Definability equals recognizability for graphs of bounded treewidth.
\newblock In {\em 2016 31st Annual ACM/IEEE Symposium on Logic in Computer
  Science (LICS)}, pages 1--10. IEEE, 2016.

\bibitem{bouchitte2001treewidth}
Vincent Bouchitt{\'e} and Ioan Todinca.
\newblock Treewidth and minimum fill-in: Grouping the minimal separators.
\newblock {\em SIAM Journal on Computing}, 31(1):212--232, 2001.

\bibitem{byskov2004enumerating}
Jesper~Makholm Byskov.
\newblock Enumerating maximal independent sets with applications to graph
  colouring.
\newblock {\em Operations Research Letters}, 32(6):547--556, 2004.

\bibitem{chang2001tree}
Maw-Shang Chang and Haiko M{\"u}ller.
\newblock On the tree-degree of graphs.
\newblock In {\em International Workshop on Graph-Theoretic Concepts in
  Computer Science}, pages 44--54. Springer, 2001.

\bibitem{chao2000optimal}
HS~Chao, Fang-Rong Hsu, and Richard C.~T. Lee.
\newblock An optimal algorithm for finding the minimum cardinality dominating
  set on permutation graphs.
\newblock {\em Discrete Applied Mathematics}, 102(3):159--173, 2000.

\bibitem{colbourn1981testing}
Charles~J Colbourn.
\newblock On testing isomorphism of permutation graphs.
\newblock {\em Networks}, 11(1):13--21, 1981.

\bibitem{coppersmith1987matrix}
Don Coppersmith and Shmuel Winograd.
\newblock Matrix multiplication via arithmetic progressions.
\newblock In {\em Proceedings of the nineteenth annual ACM symposium on Theory
  of computing}, pages 1--6, 1987.

\bibitem{corneil1981complement}
Derek~G Corneil, Helmut Lerchs, and L~Stewart Burlingham.
\newblock Complement reducible graphs.
\newblock {\em Discrete Applied Mathematics}, 3(3):163--174, 1981.

\bibitem{corneil1984clustering}
Derek~G Corneil and Yehoshua Perl.
\newblock Clustering and domination in perfect graphs.
\newblock {\em Discrete Applied Mathematics}, 9(1):27--39, 1984.

\bibitem{corneil1985linear}
Derek~G. Corneil, Yehoshua Perl, and Lorna~K Stewart.
\newblock A linear recognition algorithm for cographs.
\newblock {\em SIAM Journal on Computing}, 14(4):926--934, 1985.

\bibitem{deogun1994polynomial}
Jitender~S Deogun and George Steiner.
\newblock Polynomial algorithms for hamiltonian cycle in cocomparability
  graphs.
\newblock {\em SIAM Journal on Computing}, 23(3):520--552, 1994.

\bibitem{eppstein2001improved}
David Eppstein.
\newblock Improved algorithms for 3-coloring.
\newblock In {\em Proceedings of the twelfth annual ACM-SIAM symposium on
  Discrete algorithms}, volume 103, page 329. SIAM, 2001.

\bibitem{eppstein2002small}
David Eppstein.
\newblock Small maximal independent sets and faster exact graph coloring.
\newblock {\em J. Graph Algorithms Appl}, 7(2):131--140, 2002.

\bibitem{even1972permutation}
Shimon Even, Amir Pnueli, and Abraham Lempel.
\newblock Permutation graphs and transitive graphs.
\newblock {\em Journal of the ACM (JACM)}, 19(3):400--410, 1972.

\bibitem{fomin2008exact}
Fedor~V Fomin, Dieter Kratsch, Ioan Todinca, and Yngve Villanger.
\newblock Exact algorithms for treewidth and minimum fill-in.
\newblock {\em SIAM Journal on Computing}, 38(3):1058--1079, 2008.

\bibitem{gary1979computers}
Michael~R Gary and David~S Johnson.
\newblock Computers and {I}ntractability: A guide to the theory of
  {NP}-completeness, 1979.

\bibitem{golumbic2004algorithmic}
Martin~Charles Golumbic.
\newblock {\em Algorithmic graph theory and perfect graphs}.
\newblock Elsevier, 2004.

\bibitem{gramm2009data}
Jens Gramm, Jiong Guo, Falk H{\"u}ffner, and Rolf Niedermeier.
\newblock Data reduction and exact algorithms for clique cover.
\newblock {\em Journal of Experimental Algorithmics (JEA)}, 13:2--2, 2009.

\bibitem{hujtera1993number}
Mih{\'a}ly Hujtera and Zsolt Tuza.
\newblock The number of maximal independent sets in triangle-free graphs.
\newblock {\em SIAM Journal on Discrete Mathematics}, 6(2):284--288, 1993.

\bibitem{kellerman1973determination}
Eduardo Kellerman.
\newblock Determination of keyword conflict.
\newblock {\em IBM Technical Disclosure Bulletin}, 16(2):544--546, 1973.

\bibitem{kloks1997treewidth}
Ton Kloks, Dieter Kratsch, and Jeremy Spinrad.
\newblock On treewidth and minimum fill-in of asteroidal triple-free graphs.
\newblock {\em Theoretical Computer Science}, 175(2):309--335, 1997.

\bibitem{koivisto20062}
Mikko Koivisto.
\newblock An ${O}^*(2^{n} )$ algorithm for graph coloring and other
  partitioning problems via inclusion--exclusion.
\newblock In {\em 2006 47th Annual IEEE Symposium on Foundations of Computer
  Science (FOCS'06)}, pages 583--590. IEEE, 2006.

\bibitem{lawler1976note}
Eugene~L Lawler and LAWLER EL.
\newblock A note on the complexity of the chromatic number problem.
\newblock 1976.

\bibitem{maniu2019experimental}
Silviu Maniu, Pierre Senellart, and Suraj Jog.
\newblock An experimental study of the treewidth of real-world graph data.
\newblock In {\em 22nd International Conference on Database Theory (ICDT
  2019)}. Schloss Dagstuhl-Leibniz-Zentrum fuer Informatik, 2019.

\bibitem{mcconnell1999modular}
Ross~M McConnell and Jeremy~P Spinrad.
\newblock Modular decomposition and transitive orientation.
\newblock {\em Discrete Mathematics}, 201(1-3):189--241, 1999.

\bibitem{meister2010treewidth}
Daniel Meister.
\newblock Treewidth and minimum fill-in on permutation graphs in linear time.
\newblock {\em Theoretical computer science}, 411(40-42):3685--3700, 2010.

\bibitem{mohring1985algorithmic}
Rolf~H M{\"o}hring.
\newblock Algorithmic aspects of comparability graphs and interval graphs.
\newblock In {\em Graphs and Order}, pages 41--101. Springer, 1985.

\bibitem{mondal2003optimal}
Sukumar Mondal, Madhumangal Pal, and Tapan~K Pal.
\newblock An optimal algorithm to solve the all-pairs shortest paths problem on
  permutation graphs.
\newblock {\em Journal of Mathematical Modelling and Algorithms}, 2(1):57--65,
  2003.

\bibitem{moon1965cliques}
John~W Moon and Leo Moser.
\newblock On cliques in graphs.
\newblock {\em Israel journal of Mathematics}, 3(1):23--28, 1965.

\bibitem{orlin1977contentment}
James Orlin et~al.
\newblock Contentment in graph theory: covering graphs with cliques.
\newblock In {\em Indagationes Mathematicae (Proceedings)}, volume~80, pages
  406--424. North-Holland, 1977.

\bibitem{parra1997characterizations}
Andreas Parra and Petra Scheffler.
\newblock Characterizations and algorithmic applications of chordal graph
  embeddings.
\newblock {\em Discrete Applied Mathematics}, 79(1-3):171--188, 1997.

\bibitem{pnueli1971transitive}
Amir Pnueli, Abraham Lempel, and Shimon Even.
\newblock Transitive orientation of graphs and identification of permutation
  graphs.
\newblock {\em Canadian Journal of Mathematics}, 23(1):160--175, 1971.

\bibitem{robertson1983graph}
Neil Robertson and Paul~D Seymour.
\newblock Graph minors. {I}. excluding a forest.
\newblock {\em Journal of Combinatorial Theory, Series B}, 35(1):39--61, 1983.

\bibitem{robertson1986graph}
Neil Robertson and Paul~D. Seymour.
\newblock Graph minors. {II}. algorithmic aspects of tree-width.
\newblock {\em Journal of algorithms}, 7(3):309--322, 1986.

\bibitem{spinrad1985comparability}
Jeremy Spinrad.
\newblock On comparability and permutation graphs.
\newblock {\em SIAM Journal on Computing}, 14(3):658--670, 1985.

\bibitem{supowit1985decomposing}
Kenneth~J Supowit.
\newblock Decomposing a set of points into chains, with applications to
  permutation and circle graphs.
\newblock {\em Information Processing Letters}, 21(5):249--252, 1985.

\bibitem{tsukiyama1977new}
Shuji Tsukiyama, Mikio Ide, Hiromu Ariyoshi, and Isao Shirakawa.
\newblock A new algorithm for generating all the maximal independent sets.
\newblock {\em SIAM Journal on Computing}, 6(3):505--517, 1977.

\end{thebibliography}
\bibliographystyle{plain}

\end{document}